\documentclass[%
 preprint,
 amsmath,amssymb,
 aps,
]{revtex4-2}

\usepackage{graphicx}
\usepackage{dcolumn}
\usepackage{bm}
\usepackage{dsfont}

\usepackage{algorithm}
\usepackage[noend]{algpseudocode}
\usepackage{amsthm}
\newcommand{\pluseq}{\mathrel{+}=}

\newtheorem{theorem}{Theorem}
\newtheorem{lemma}{Lemma}
\usepackage{enumitem}

\usepackage{tikz}
\usetikzlibrary{quantikz}
\usepackage{yquant}
\usepackage{braket}

\usepackage{listings}
\usepackage{courier}
\usepackage{float}

\begin{document}


\title{Partial oracles quantum algorithm framework\\
Part I: Analysis of in-place operations}

\author{Fintan Bolton}
\email{fintan.bolton@bradan-quantum.com}
\affiliation{Bradan Quantum, Klugstr. 101, 80637 Munich, Germany}



\date{\today}

\begin{abstract}


The partial oracles framework is a quantum search algorithm that has the potential to exceed the quadratic speedup of Grover's algorithm, up to a theoretical maximum of an exponential speedup.
Until now, however, the framework has lacked an explicit method for constructing the operator that represents the search iteration.
In this paper, we provide the missing construction, for the special case of an oracle function definable using only in-place operations (that is, where the calculated result of the oracle function can be read just from the qubits in the search index).
The restriction to in-place operations means that the current work does not yet exhibit quantum advantage: oracle functions constructed using only in-place operations are always classically reversible.
To demonstrate quantum advantage, it will be necessary to extend this construction method to include out-of-place operations (part II).
As part of the construction of the search iteration operator, we define a new type of transform, the \textit{reciprocal transform}, which is applied to the oracle function.
We show that the reciprocal transform obeys a chain rule, which makes it possible to break down complex transforms into simple steps.
To illustrate the practical application of this search method, we apply the reciprocal transform to elementary operations from the SHA-256 hash algorithm: addition modulo $2^n$, the $Maj(a, b, c)$ function, the $Ch(a, b, c)$ function, and the bit shift functions.
We also introduce the QFrame python library, which is used to automate the construction of quantum circuits that represent reciprocal transforms.

\end{abstract}

\maketitle


\section{Introduction}

The number of fundamental quantum algorithms we have at our disposal is relatively small and many of these are tied to solving a specific problem (for example, Shor's algorithm for prime number factorization).
Among these algorithms, however, Grover's search algorithm stands out for its unmatched flexibility, and many applications of quantum computing rely on it as a fundamental building block to deliver quantum advantage.
But the quadratic computational speedup offered by Grover's algorithm (locating a solution in a search space of size $N$ using only $\mathcal{O}(\sqrt{N})$ oracle queries) might not actually be of practical use in the real world.

In a recent paper \cite{hoefler2023disentanglinghypepracticalityrealistically}, Hoefler, Häner and Troyer made a performance comparison between a contemporary GPU processor (state-of-the art as of 2023) and a highly optimistic future projection of a quantum processor with 10,000 error-corrected qubits.
Taking a two-week computation time as the cutoff to define a limit on the size of the problem considered, they found that a quadratic speedup was of no practical use (at the threshhold of quantum advantage, an oracle function would be limited to performing less than 0.2 of a single 16-bit floating point operation).
Stoudenmire and Waintal \cite{Stoudenmire2023GroversAO,PhysRevX.14.041029} have also made a detailed study of the practicality of Grover's algorithm and come to similar conclusions.

Both of these studies underscore the importance of developing quantum algorithms that can deliver an exponential speedup, as this is the class of algorithms that can be relied upon to deliver a decisive quantum advantage.

The partial oracles algorithm (still under development) is a line of research that was devised specifically to tackle the limitations of Grover search.
In a previous paper \cite{PartialOraclesArticle2024}, we demonstrated that the theoretical upper limit (not necessarily achievable) on the speedup provided by the partial oracles approach is exponential.
This does not contradict the well-known result that Grover's search is optimal \cite{Galindo2000FamilyOG}, because the partial oracles approach starts from a different premise: instead of a single-bit oracle function $f(x):\{0,1\}^n\rightarrow\{0,1\}$ which returns $0$ for the solution value $x_s$ (zero-matching convention), the partial oracles algorithm employs a multi-bit oracle function $f(x):\{0,1\}^n\rightarrow\{0,1\}^n$ which indicates a full match only if every returned bit is zero----that is, $f_0(x)=0, f_1(x)=0,\ldots,f_{n-1}(x)=0$ indicates the solution has been found (assuming a zero-matching convention).

The particular difficulty faced by the partial oracles approach, however, is that the the second Grover operator can no longer be used to perform intermediate iterations of the search and a new type of operator is required.
Consider how the Grover operators are conventionally defined.
Given an $n$-qubit index register $\ket{x}=\ket{x_{n-1}\ldots x_0}$, we define the initial state $\ket{S}=(1/\sqrt{2^n})\sum_x \ket{x}$ and the sought-after solution state $\ket{x_s}$ that satisfies $f(x_s)=0$.
We can then define the \textit{first Grover operator} $G^I=\mathbf{I} - 2\ket{x_s}\bra{x_s}$ as a reflection in the hyperplane perpendicular to $\ket{x_s}$; and the \textit{second Grover operator} $G^{II}=\mathbf{I} - 2\ket{S}\bra{S}$ as a reflection in the hyperplane perpendicular to the initial state $\ket{S}$.
A single Grover iteration is then defined as the product $G^{II}G^I$ of these two reflection operators:

\begin{equation}
- \big( \mathbf{I} - 2\ket{S}\bra{S} \big)\,\, \big( \mathbf{I} - 2\ket{x_s}\bra{x_s} \big)
\end{equation}

When we consider that the initial state $\ket{S}$ can be written as $H^{\otimes n}\,\ket{0}$ (applying a Walsh-Hadamard transformation to the $n$-qubit zero state) and the fact that $H^{\otimes n}$ is its own inverse, we can also write the Grover iteration in the form:

\begin{equation}
- \underbrace{H^{\otimes n}\,\big( \mathbf{I} - 2\ket{0}\bra{0} \big)\, H^{\otimes n}}_{\text{2\textsuperscript{nd} Grover operator}}
\,\,
\underbrace{\big( \mathbf{I} - 2\ket{x_s}\bra{x_s} \big)}_{\textrm{1\textsuperscript{st} Grover operator}}
\end{equation}

From this perspective, we can break the 2\textsuperscript{nd} Grover operator into three parts:

\begin{enumerate}
\item The Walsh-Hadamard transformation $H^{\otimes n}$ maps states from direct space $\ket{x}$ to reciprocal space $\ket{k}$, where the basis vectors $\ket{k}$ represent the Walsh functions $\{W_k(x)\}$ for $k=0\ldots (2^n-1)$.

\item In reciprocal space, the operator performs a reflection in the hyperplane perpendicular to the $\ket{k=0}$ reciprocal state (which represents the $W_0(x)$ Walsh function).

\item The Walsh-Hadamard transformation $H^{\otimes n}$ maps states from reciprocal space $\ket{k}$ back to direct space $\ket{x}$.
\end{enumerate}

As already noted, we have found that the second Grover operator cannot be used in the context of the partial oracles approach.
But it turns out that it is possible to define an operator in reciprocal space (sandwiched between the two $H^{\otimes n}$ operators) that implements the appropriate type of search iteration for the partial oracles approach.
The requisite operator is constructed by performing what we call a \textit{reciprocal transform} $R[f^\sigma (x^\mu)]$ of the oracle function (section \ref{reciprocaltransformdefn}).
The corresponding partial oracle iteration operator is then given by equation \ref{eq:partialoracleiteration-sequential} (or in parallelized form, by equation \ref{eq:partialoracleiteration-single}).
This is the key result of this paper.

Note, however, that this result has so far only been proved for in-place operations, not out-of-place operations.
Given an $n$-qubit index register $x$, the distinction between in-place operations and out-of-place operations can be described, as follows:

\begin{itemize}
\item \textit{In-place operations}---the oracle function can be computed in such a way that it is only necessary to test qubits $x_{n-1},\ldots,x_0$ from the index register $x$ when applying a phase marker.

\item \textit{Out-of-place operations}---the oracle function must be computed using out-of-place qubits (that is, not part of the index $x$) in such a way that it is necessary to test out-of-place qubits when applying a phase marker.
\end{itemize}

The difference between these two cases is stark: an oracle function calculated using in-place operations is always classically reversible; whereas reversing an oracle function calculated using out-of-place operations can be very non-trivial and hard to compute.

A typical example of an out-of-place operation is integer multiplication.
Given two index registers $x$ and $y$ (with $x$ spanning the qubits $x_{\ell-1}\ldots x_0$, and $y$ spanning the qubits $y_{m-1}\ldots y_0$), the product of $x$ and $y$ \textbf{must} be calculated in an out-of-place register $z$, as there is no algorithm for in-place multiplication.
Consequently, solving the factorization problem requires support for out-of-place operations.

So far, the partial oracles algorithm has been worked out only for in-place operations and therefore does not (yet) demonstrate quantum advantage.
The next step in the development of the partial oracles algorithm is the analysis of out-of-place operations, which will be the subject of a future paper (part II).

\section{Motivation}

The ancestor of all quantum search algorithms is Grover's famous algorithm \cite{LKGrover1996,PhysRevLett.79.325}, which also serves as the starting point for the partial oracles algorithm.
In the usual formulation, we are given an index register $x$, which consists of $n$ qubits $x=x_{n-1}x_{n-1}\ldots x_1 x_0$ and is capable of representing $N=2^n$ distinct states $x=\{0,1\}^n$ (which represent the search space).
We are also given an oracle function $f(x)$ that identifies the sought-after solution state $x_s$.
The usual convention is for the oracle function to return $f(x_s)=1$ for the solution state $x_s$ and $f(x)=0$ for all other states $x\neq x_s$---but for the purpose of this work, we find it more useful to adopt the convention that $f(x_s)=0$ for the solution state $x_s$ and $f(x)=1$ for all other states (zero matching convention).

Grover's algorithm consists of the following steps:

\begin{enumerate}
\item The quantum system is initialized in an equal superposition state $\ket{\Psi} = (1/N^{1/2})\,\sum_x\,\ket{x}$, where each $\ket{x}$ state has the amplitude $+1/\sqrt{N}$.

\item The first Grover operator uses the oracle function $f(x)$ to identify the solution state $\ket{x_s}$ and marks it by applying the phase $\exp(i\pi)=-1$.

\item A Walsh-Hadamard transformation is applied to the state $\ket{\Psi}\mapsto\,H^{\otimes n}\,\ket{\Psi}$, which transforms the state into reciprocal space.
In reciprocal space, the basis states are given by the Walsh functions $W_k(x)=\exp(i\,2\pi\,k\cdot x)$, where we define $k\cdot x= (1/2)\sum_{j=1}^{n-1}k_j x_j$. Most of the state's amplitude is concentrated in the zero Walsh state $W_0(x) = \textrm{const}$, while a small part of the amplitude belongs to the negative amplitude at $x_s$, representing the solution state.

\item The second Grover operator flips the zero Walsh state $W_0(x)$ (for example, using a multi-controlled gate to detect the all-zero state, and an ancillary phase oracle to flip the phase), so that its amplitude is negative and the components it represents are oriented parallel to the negative amplitude at $x_s$.

\item The (inverse) Walsh-Hadamard transformation is applied to the state, transforming from reciprocal space back to direct space.
After this transformation, the amplitude from the zero Walsh function and the solution amplitude add constructively at $x_s$, so that the ampltitude at $x_s$ increases from $1/\sqrt{N}$ to approximately $2/\sqrt{N}$.

\item By repeating this procedure iteratively, the amplitude at $x_s$ increases---at first linearly and then more slowly---until all of the amplitude is in the $\ket{x_s}$ state, after approximately $\sqrt{N}$ iterations.
\end{enumerate}

In the partial oracle approach \cite{PartialOraclesArticle2024}, the single-bit oracle function $f(x):\{0,1\}^n\rightarrow\{0,1\}$ is replaced by a multi-bit oracle function $f(x):\{0,1\}^n\rightarrow\{0,1\}^n$.
Instead of returning just one bit (indicating either match or no-match), the partial oracle function returns multiple bits $\{f_j(x)\}_{j=0}^{n-1}$, where each returned bit $f_j$ indicates a \textit{partial match}.
When all of the returned bits signal a match---that is, when simultaneously $f_0(x)=0, f_1(x)=0,\ldots,f_{n-1}(x)=0$ (assuming a zero-matching convention)---this indicates that the solution value $x_s$ has been found.

\begin{figure}[h]
\centering
\includegraphics[width=\linewidth]{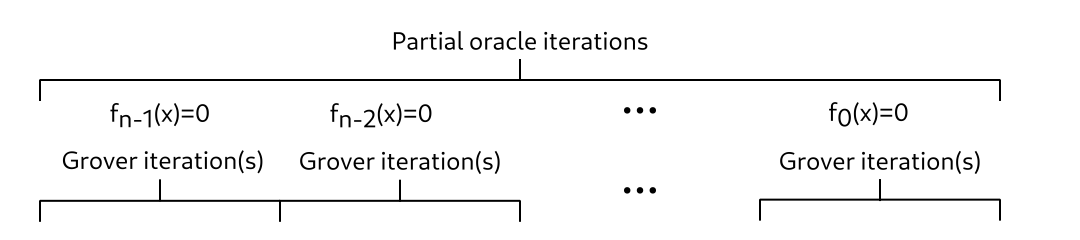}
\caption{Overview of partial oracle iterations. At each iteration stage $j$, a search is executed to return the index values $x$ that satisfy the corresponding partial oracle condition $f_j(x)=0$, and the effect of these searches is cumulative, so that at the end of the algorithm the remaining index state (or states) satisfies all of the partial oracle conditions, $\{f_j(x)=0\}_{j=0}^{n-1}$. Each of these $n$ searches is implemented by a (modified) Grover-Long iteration, which is iterated until the $f_j(x)=0$ condition is satisfied. In the simplest case, when the condition $f_j(x)=0$ is satisfied by 50\% of index states, the modified Grover-Long search requires only one iteration.}
\label{fig:partialoracleiterations}
\end{figure}

As shown in figure Figure \ref{fig:partialoracleiterations}, the partial oracles algorithm is built on two layers of iteration, with Grover's algorithm nested inside the partial oracle iterations.
In the first partial oracle iteration, we try to find the subset of index values $x$ that satisfy the condition $f_{n-1}(x)=0$.
For this, we employ Grover's algorithm, iterating until the condition $f_{n-1}(x)=0$ is satisfied for all $x$.
We then move on to the next partial oracle iteration, which requires both the conditions $f_{n-1}(x)=0$ and $f_{n-2}(x)=0$ (or equivalently, $f_{n-1}(x)\vee f_{n-2}(x)=0$, where $\vee$ is logical OR).
We continue in this way across all of the partial oracle iterations until, at the last iteration, we find the solution value $x_s$ that satisfies all of the partial conditions $f_{n-1}(x)\vee\cdots\vee f_0(x)=0$.
This is the general scheme of the algorithm, but there are significant difficulties that need to be addressed: from the second partial oracle iteration onward, it becomes impossible to apply the conventional Grover operator and a new kind of operator is required.

Consider the special case of a partial oracle function $f(x):\{0,1\}^n\rightarrow\{0,1\}^n$ that is bijective, mapping each index value $x=x_{n-1}\ldots x_0$ to a unique multi-bit flag value $f=f_{n-1}\ldots f_0$.
When $f(x)$ is bijective, it can easily be shown that each flag bit $f_j$ divides the set of index values $x$ into two equal-sized subsets: that is, the set of all $x$ satisfying $f_j(x)=0$ has size $N/2$ and is disjoint from the set of all $x$ satisfying $f_j(x)=1$, which also has size $N/2$.
This implies that, for the first partial oracle iteration, we start with a set of $N$ index values and finish with an (intermediate) solution set of $N/2$ index values satisfying $f_{n-1}(x)=0$.
For this iteration, which requires a small number of Grover iterations (exactly one iteration), it is better to use the Grover-Long variant of the algorithm \cite{Long2001GroverAW}, which steers the Grover rotation more precisely and more quickly to its target.

\begin{figure}[h]
\centering
\includegraphics[width=\linewidth*4/10]{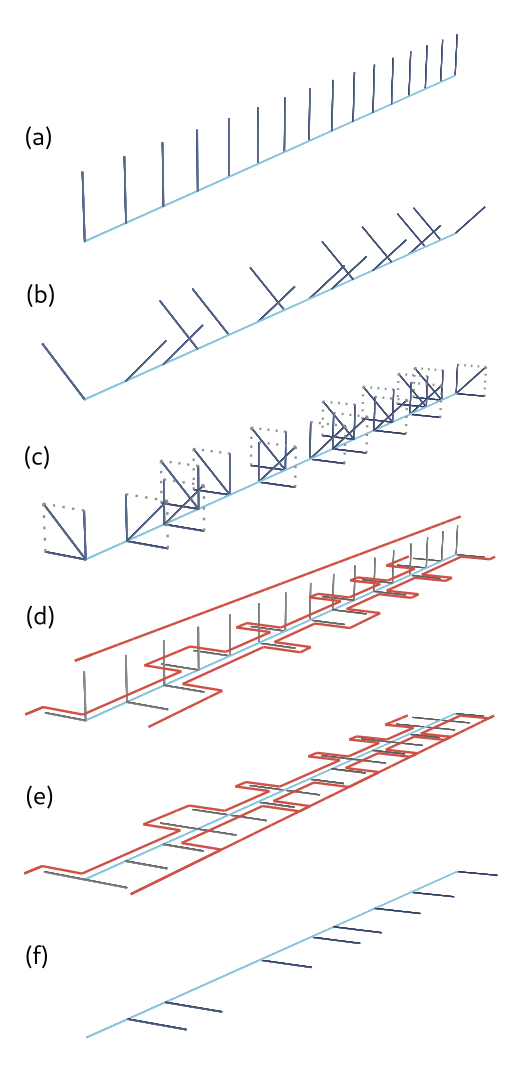}
\caption{First partial oracle iteration showing: (a) the initial equal superposition state; (b) index states after applying a relative phase of $\exp(i\pi/2)$ between matching and non-matching states; (c) horizontal and vertical components of index states; (d) visualization of states in reciprocal space, after Walsh-Hadamard transform; (e) after applying the phase $\exp(i\pi/2)$ to the $W_0(x)$ Walsh state; (f) after transforming from reciprocal space back to direct space.}
\label{fig:firstpartialoracleiteration}
\end{figure}

Figure \ref{fig:firstpartialoracleiteration} shows a visualization of the Grover-Long algorithm, implementing the first partial oracle iteration (which targets the condition $f_{n-1}(x)=0$):

\begin{enumerate}[label=(\alph*)]
\item The quantum system is initialized in an equal superposition state $\ket{\Psi} = (1/N^{1/2})\,\sum_x\,\ket{x}$, where each $\ket{x}$ state has the amplitude $+1/\sqrt{N}$.

\item The Grover-Long algorithm uses the partial oracle $f_{n-1}(x)$ to identify the index states $\ket{x}$ that match the condition $f_{n-1}(x)=0$, applying a relative phase of $\exp(i\,\pi/2)$ between the matching and non-matching states.
In Figure \ref{fig:firstpartialoracleiteration} (b), this is visualized by rotating the non-matching states anti-clockwise through the angle $\pi/4$ and the matching states clockwise through the angle $\pi/4$.

\item If we now consider the magnitude of the component projected onto the vertical axis, we see that each $\ket{x}$ has a vertical component equal to $1/\sqrt{2}$ times its amplitude.
The horizontal component also has a magnitude equal to $1/\sqrt{2}$ times the amplitude, pointing to the left for the non-matching states and to the right for the matching states.

\item We now apply a Walsh-Hadamard transformation to the state, which transforms the state into reciprocal space.
Exactly half of the state's amplitude is assigned to the zero Walsh state $W_0(x)$, which represents the vertical components of all the index states $\ket{x}$.
The rest of the amplitude is shared across a complex superposition of Walsh functions, representing the series of spikes that point left and right along the horizontal axis.

\item The second Grover-Long operator applies the phase $\exp(i\,\pi/2)$ to the zero Walsh state $W_0(x)$, so that its amplitude is oriented parallel to the horizontal components of the matching $\ket{x}$ states and anti-parallel to the non-matching $\ket{x}$ states.

\item We now invert the Walsh-Hadamard transformation, transforming from reciprocal space back to direct space.
After this transformation, the amplitude from the zero Walsh function and the oracle-matching horizontal components add constructively (doubling their amplitudes from $1/\sqrt{2}$ to $2/\sqrt{2}$); while the amplitude from the zero Walsh function and the non-oracle-matching horizontal components add destructively (reducing their amplitudes to zero).
\end{enumerate}

The second partial oracle iteration (targeting $f_{n-1}(x)=0$ and $f_{n-2}(x)=0$) reveals a new difficulty, however.
When you consider the state of the system after the first partial oracle iteration, the non-matching states that have just been eliminated are typically removed from quasi-random locations along the continuum of $\ket{x}$ states.
We can anticipate that, after transforming to reciprocal space, this state yields a quasi-random superposition of Walsh functions $\sum_k c_k\,W_k$ that cannot be solved using the Grover-Long algorithm.

However, there is a way around this problem.
It turns out that, after applying the Walsh-Hadamard transformation, it is possible to reorder the Walsh functions in reciprocal space, yielding a compact representation of the vertical components of the index states, which can then be addressed by a single reciprocal qubit.
The operator that performs the reordering of Walsh functions in reciprocal space is \textit{the reciprocal transform of the oracle function} (see section \ref{reciprocaltransformdefn}).
This represents the key result of this paper.

\begin{figure}
\centering
\includegraphics[width=\linewidth*4/10]{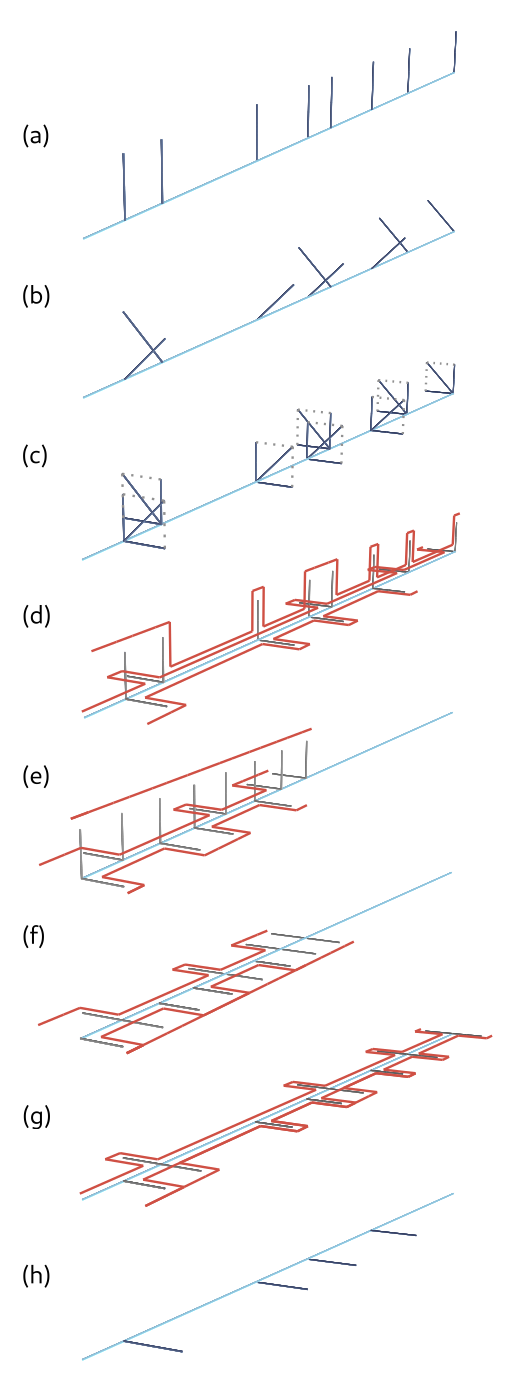}
\caption{Second (and subsequent) partial oracle iteration: (a) index states remaining after the first partial oracle iteration; (b) index states after applying a relative phase of $\exp(i\pi/2)$ between matching and non-matching states; (c) horizontal and vertical components of index states; (d) visualization of states in reciprocal space, after Walsh-Hadamard transform; (e) after being rearranged by the reciprocal transform; (f) after applying the phase $\exp(i\pi/2)$ to the $W_0(x)$ Walsh state; (g) after being rearranged by the adjoint of the reciprocal transform; (h) after transforming from reciprocal space back to direct space.}
\label{fig:secondpartialoracleiteration}
\end{figure}

Figure \ref{fig:secondpartialoracleiteration} shows a visualization of the modified Grover-Long algorithm, implementing the second partial oracle iteration and incorporating additional steps using the reciprocal transform of the oracle function:

\begin{enumerate}[label=(\alph*)]
\item The initial state of the quantum system (after the first partial oracle iteration) is a quasi-randum distribution of states, where 50\% of the states are missing from the continuum of $\ket{x}$ states.

\item The Grover-Long algorithm uses the partial oracle $f_{n-2}(x)$ to identify the index states $\ket{x}$ that match the condition $f_{n-2}(x)=0$, applying a relative phase of $\exp(i\,\pi/2)$ between the matching and non-matching states.
In (b), this is visualized by rotating the non-matching states anti-clockwise through the angle $\pi/4$ and the matching states clockwise through the angle $\pi/4$.

\item In this second partial oracle iteration, the components of $\ket{x}$ projected onto the vertical axis are just as sparse and randomly distributed as the components projected onto the horizontal axis.

\item We apply a Walsh-Hadamard transformation to the state, which transforms the state into reciprocal space.
In this case, both the vertical components and the horizontal components are represented by a complex (quasi-random) superposition of Walsh functions $\{W_k\}$.

\item We apply the operator $R[f(x)]$, which represents the reciprocal transform of the oracle function (see section \ref{reciprocaltransformdefn}).
The effect of this operator is to reorder the Walsh functions in reciprocal space, giving a compact representation where 50\% of the amplitude is assigned to the zero Walsh state $W_0(x)$ (representing the vertical components of all the index states).

\item We apply the phase $\exp(i\,\pi/2)$ to the zero Walsh state $W_0(x)$, so that its amplitude is oriented parallel to the horizontal components of the matching $\ket{x}$ states and anti-parallel to the non-matching $\ket{x}$ states.

\item We apply the operator $R^\dagger[f(x)]$, which represents the \textit{inverse} reciprocal transform of the oracle function (see section \ref{reciprocaltransformdefn}).
The effect of this operator is to undo the reordering of the Walsh functions in reciprocal space.

\item We invert the Walsh-Hadamard transformation, transforming from reciprocal space back to direct space.
After this transformation, the amplitudes from the rotated components and the oracle-matching horizontal components add constructively (doubling their amplitudes from $1/\sqrt{2}$ to $2/\sqrt{2}$); while the amplitudes from the rotated components and the non-oracle-matching horizontal components add destructively (reducing their amplitudes to zero).

\end{enumerate}

Continuing in this way for the third and subsequent partial oracle iterations, we reduce the number of states by $1/2$ at each iteration, progressively reducing the number of basis states $\ket{x}$ from $2^n\mapsto 2^{n-1}\mapsto\cdots\mapsto 2\mapsto 1$ until we obtain the solution state $\ket{x_s}$ after $n$ iterations.

\section{Partial oracles algorithm framework}

\subsection{One-to-one oracle function} \label{onetooneoracle}

The partial oracle framework works with a multi-qubit oracle function. That is, given an $n$-qubit index register $x\in\{0,1\}^n$, the oracle function $f(x):\{0, 1\}^n\rightarrow\{0, 1\}^n$ returns the $n$-qubit value $f(x)=(f_0(x), \ldots, f_{n-1}(x))$.

We use a subscript to reference the qubits in a register. The $x$ index register therefore has the qubits $x_0,\ldots,x_{n-1}$ and the $f(x)$ oracle function returns the qubits $f_0,\ldots,f_{n-1}$. We use the notation $f_j(x)$ to denote the \textit{partial oracle function} that returns the $j$th qubit of $f(x)$.

There are two alternative conventions we can use to define the search match condition:

\begin{itemize}
\item \textit{Match on all zeros}---where an index value $x$ is defined to belong to the solution set, if $f_j(x)=0$ for $j=0,\ldots,n-1$. That is, we say $x$ is a matching index when $f(x)$ satisfies: $|f_{n-1}\ldots f_0\rangle = |0\ldots0\rangle$.

\item \textit{Match on all ones}---where an index value $x$ is defined to belong to the solution set, if $f_j(x)=1$ for $j=0,\ldots,n-1$. That is, we say $x$ is a matching index when $f(x)$ satisfies: $|f_{n-1}\ldots f_0\rangle = |1\ldots1\rangle$.

\end{itemize}

Both conventions are potentially useful, depending on the context. For the examples in this paper, we stick to the \textit{match on all zeros} convention.

For the partial oracle framework described in this paper we require that the oracle function $f(x)$ is defined explicitly as a one-to-one function. That is, for each $n$-qubit index $x$, the oracle function returns a unique $n$-qubit value. In practice, however, you will often find that a function you want to use as an oracle does not have an explicitly one-to-one form. For example, consider the majority function $Maj(a, b, c)$, which is used in the definition of SHA algorithms from the secure hash standard. The majority function takes three bits $a, b, c$ as its arguments and is defined as follows (where $\wedge$ denotes logical AND, and $\oplus$ denotes logical XOR):

\begin{equation}
Maj(a, b, c) = (a\wedge b) \oplus (b\wedge c) \oplus (c\wedge a)
\end{equation}

This function returns $0$ if the majority of its inputs $a, b, c$ are $0$ and returns $1$ if the majority of its inputs are $1$. Consider how to define an oracle function that selects for all of the triplets $(a, b, c)$ satisfying $Maj(a, b, c)=1$. We could define this oracle function as follows (matching on zero convention):

\begin{equation}
f(a, b, c) = Maj(a, b, c)\oplus 1
\end{equation}

The problem with this definition, however, is that it maps 3 bits to 1 bit, and is therefore not a one-to-one (bijective) function. This can be fixed by \textit{completing} the oracle function, expanding it to a multi-bit function that returns three bits instead of one (giving a mapping of the form $f(x):\{0, 1\}^3\rightarrow\{0, 1\}^3$). A complete oracle function for the majority function can be defined as follows:

\begin{align}
f_0 &= Maj(a, b, c)\oplus 1 \\
f_1 &= a\oplus c \\
f_2 &= a\oplus b
\end{align}

Where it can easily be verified that the resulting mapping $(a, b, c)\mapsto (f_0, f_1, f_2)$ is one-to-one. That is, for every triplet $(a, b, c)$ there is a unique corresponding triplet $(f_0, f_1, f_2)$.

Note that bitwise one-to-one oracle functions have an important property when it comes to matching conditions: \textit{each matching partial oracle bit divides the index values into two equally large subsets}. For example, the condition $f_0(x)=0$ is satisified by exactly half of the index values $x$ and the condition $f_0(x)=1$ is satisified by the other half (giving two disjoint solution sets). Taking the majority function as an example: the condition $f_0=0$ is satisfied by 4 index values, $f_0\vee f_1=0$ is satisified by 2 index values, and $f_0\vee f_1\vee f_2=0$ is satisfied by 1 index value.

\subsection{Notation and definition of qubit shapes}

In the course of developing the partial oracle framework, we find it is necessary to work with multiple qubits and qubit registers at the same time. It is, therefore, helpful to introduce some consistent notation for these objects.

\subsubsection{Registers in direct space}

The raw register $x\in\{0,1\}^n$ refers to the totality of qubits in the search index (but does not include, for example, any ancillary qubits that might be needed by circuit implementations). We then divide the $n$ qubits of the raw register $x$ into $r$ registers of sizes $\{m_j\}_{j=1}^{r}$., where $m_1+m_2+\cdots m_r=n$. To identify each of the $r$ registers, we use a superscript index in parentheses and a subscript to reference individual qubits:

\begin{align*}
1^{\textrm{st}} \> \textrm{register with} \> m_1 \> \textrm{qubits} &= x^{(1)}
= x^{(1)}_{m_1-1} x^{(1)}_{m_1-2} \ldots x^{(1)}_0 \\
\vdots & \\
r^{\textrm{th}} \> \textrm{register with} \> m_r \> \textrm{qubits} &= x^{(r)}
= x^{(r)}_{m_r-1} x^{(1)}_{m_r-2} \ldots x^{(r)}_0 \\
\end{align*}

The collection of registers constitutes a \textit{qubit vector}:

\begin{equation*}
x^\mu = (x^{(1)}, x^{(2)},\ldots,x^{(r)})
\end{equation*}

Where the greek superscript $\mu$ denotes the \textit{qubit shape} of the qubit vector:

\begin{equation*}
\mu = (m_1, m_2,\ldots,m_r)
\end{equation*}

\subsubsection{Registers in reciprocal space}

By analogy to the notation in direct space, we can also define the corresponding register notation in reciprocal space. Given the \textit{reciprocal raw register} $k$, we can define the corresponding \textit{reciprocal qubit vector} $k^\mu$:

\begin{equation*}
k^\mu = (k^{(1)}, k^{(2)},\ldots,k^{(r)})
\end{equation*}

With the corresponding qubit shape:

\begin{equation*}
\mu = (m_1, m_2,\ldots,m_r)
\end{equation*}

\subsubsection{Hadamard-style dot product}

We define the Hadamard-style dot product between a qubit vector $x^\mu$ and a reciprocal qubit vector $k^\mu$ as:

\begin{eqnarray}
k^\mu \cdot x^\mu &=& \left( \frac{1}{2}\bigoplus_{\ell_1=0}^{m_1-1} k^{(1)}_{\ell_1}x^{(1)}_{\ell_1} \right) \oplus \cdots \oplus \left( \frac{1}{2}\bigoplus_{\ell_r=0}^{m_r-1} k^{(r)}_{\ell_r}x^{(r)}_{\ell_r} \right) \\
 &=& \frac{1}{2} \bigoplus_{j=1}^{r} \bigoplus_{\ell_j=1}^{r} k^{(j)}_{\ell_j} x^{(j)}_{\ell_j} \\
 &=& \frac{1}{2} \bigoplus_{i=0}^{n-1} k_{i} x_{i} \label{eq:hadamarddotproduct}
\end{eqnarray}

Equation \ref{eq:hadamarddotproduct} uses the indices of the underlying raw vectors, which span all of the qubits from registers $1\ldots r$, combining the qubit indices into a single range from $0\ldots n-1$. The Hadamard dot product is useful when using the Walsh-Hadamard transformation for the transformation to reciprocal space.

\subsubsection{Hadamard transform with qubit shapes}

With the help of the dot product notation from the previous section, we can write the Hadamard transform for qubit vectors as:

\begin{equation}
H_{k^\mu x^\mu} = \frac{1}{\sqrt{N_\mu}} \sum_{k^\mu} \sum_{x^\mu} e^{-i2\pi\,k^\mu\cdot x^\mu}\,\,|k^\mu\rangle\langle x^\mu|
\end{equation}

Where $N_\mu=2^{m_1}\ldots 2^{m_r}=2^n$ and the sums over $x^\mu$ and $k^\mu$ use the following compact notation:

\begin{eqnarray}
\sum_{x^\mu} &=& \sum_{x^{(1)}=0}^{2^{m_1}-1}\cdots\sum_{x^{(r)}=0}^{2^{m_r}-1} \\ 
\sum_{k^\mu} &=& \sum_{k^{(1)}=0}^{2^{m_1}-1}\cdots\sum_{k^{(r)}=0}^{2^{m_r}-1}
\end{eqnarray}

The Dirac kets for the qubit vectors are defined in terms of the qubit register kets, as follows:

\begin{eqnarray}
|x^\mu\rangle &=& |x^{(1)}\rangle\ldots |x^{(r)}\rangle \\
|k^\mu\rangle &=& |k^{(1)}\rangle\ldots |k^{(r)}\rangle
\end{eqnarray}

\subsubsection{Delta function}

Note that, in the usual way, we can define a delta function by summing over the complete Walsh-Hadamard basis, that is:

\begin{equation}
\delta_{x^\mu, 0} = \frac{1}{N_\mu} \sum_{k^\mu} e^{-i2\pi\,k^\mu\cdot x^\mu}
\label{eq:deltafunction}
\end{equation}

This can easily be proved for the case of the Hadamard-style dot product, as follows:

\begin{eqnarray*}
\frac{1}{N_\mu} \sum_{k^\mu} e^{-i2\pi\,k^\mu\cdot x^\mu}
 &=& \left( \frac{1}{2}\sum_{k_{n-1}=0}^1 e^{-i\pi\,k_{n-1}x_{n-1}} \right)\cdots
     \left( \frac{1}{2}\sum_{k_0=0}^1 e^{-i\pi\,k_0 x_0} \right) \\
 &=& \left( \delta_{x_{n-1}, 0} \right)\cdots\left( \delta_{x_0, 0} \right) \\
 &=& \delta_{x^\mu, 0}
\end{eqnarray*}

\subsection{Reciprocal transform definition}
\label{reciprocaltransformdefn}

We are now ready to define the reciprocal transform, which plays a central role in the partial oracles algorithm. Given the $n$-dimensional one-to-one oracle function $f^\sigma(x^\mu):\{0,1\}^n\rightarrow\{0,1\}^n$, which maps qubit vector $x^\mu$ with shape $\mu$ to qubit vector $f^\sigma(x^\mu)$ with shape $\sigma$, we define the \textit{reciprocal transform} $R[f^\sigma (x^\mu)]$ of the oracle function to be:

\begin{equation}
R[f^\sigma (x^\mu)] = \frac{1}{\sqrt{N_\sigma}}\frac{1}{\sqrt{N_\mu}}
    \sum_{\kappa^\sigma}\sum_{x^\mu}\sum_{k^\mu}
    e^{i2\pi\,(-\kappa^\sigma\cdot f^\sigma(x^\mu) + x^\mu\cdot k^\mu )}
    |\kappa^\sigma\rangle\langle k^\mu|
    \label{eq:reciprocaltransform}
\end{equation}

Where the \textit{reciprocal operator} $R[f^\sigma (x^\mu)]$ is a unitary operator that acts on the $n$ qubits $|k^\mu\rangle$, returning the $n$ transformed qubits $|\kappa^\sigma\rangle$. $N_\sigma$ is the number of states in the $\sigma$ shape and $N_\mu$ is the number of states in the $\mu$ shape (which should both be equal to $2^n$ in practice). Note that $x^\mu$ is effectively a dummy variable in the above expression: after performing the sum over $x^\mu$ on the right-hand side of equation \ref{eq:reciprocaltransform}, the variable $x^\mu$ disappears from the expression.
The reciprocal operator $R[f^\sigma (x^\mu)]$ is effectively a representation of the oracle function in reciprocal space (that is, at a point in the algorithm after performing a Walsh-Hadamard transformation or a Fourier transformation on the index qubits).

Closely related to this, we can also define a bare form of the reciprocal operator $B[f^\sigma (x^\mu)]$, which goes straight from direct space $\ket{x^\mu}$ to reciprocal space $\ket{\kappa^\sigma}$, as follows:

\begin{equation}
B[f^\sigma (x^\mu)] = \frac{1}{\sqrt{N_\sigma}}
    \sum_{\kappa^\sigma}\sum_{x^\mu}
    e^{i2\pi\,(-\kappa^\sigma\cdot f^\sigma(x^\mu))}
    |\kappa^\sigma\rangle\langle x^\mu|
    \label{eq:baretransform}
\end{equation}

The reciprocal operator can be recovered from this form by preceding it with a Walsh-Hadamard transformation---that is, $R[f^\sigma (x^\mu)]=B[f^\sigma (x^\mu)]\,H^{\otimes n}$.
We make use of the bare form $B$ in the proof of the partial oracles algorithm (section \ref{partialoraclealgorithm}).
But for practical calculations, the reciprocal operator $R[f^\sigma (x^\mu)]$ is far more important, as it is straightforward to calculate it (by summing over the dummy variable $x^\mu$), whereas there is no obvious procedure for calculating $B[f^\sigma (x^\mu)]$.

\subsection{Partial oracles algorithm}
\label{partialoraclealgorithm}

Given an $n$-qubit index register $x\in\{0,1\}^n$ and a bijective partial oracle function $f(x):\{0,1\}^n\rightarrow\{0,1\}^n$, the goal of the partial oracles algorithm is to find the target index value $x_t$ that satisfies all of the partial oracle conditions: $f_0(x_t)=0,\ldots, f_{n-1}(x_t)=0$ (following the zero matching convention).

The problem can be solved by starting with a uniform superposition state:

\begin{equation}
\ket{S} = H^{\otimes n}\,\ket{0} = \frac{1}{\sqrt{2^n}}\,\sum_x\,\ket{x}
\end{equation}

And then successively applying the \textit{partial oracle iteration} for each $\ell=0,\ldots,n-1$:

\begin{equation}
H^{\otimes n}\,R^\dagger[f(x)]\,e^{+i\frac{\pi}{2}\kappa_\ell}\,R[f(x)]\,H^{\otimes n}\,e^{+i\frac{\pi}{2} f_\ell(x)}
\label{eq:partialoracleiteration-sequential}
\end{equation}

Where the $\exp(+i\frac{\pi}{2}\kappa_\ell)$ operator references the $\kappa_\ell$ qubit in the state returned by the reciprocal operator $R[f(x)]$ (see equation \ref{eq:reciprocaltransform}).

We now proceed to show how the partial oracle iterations reduce the uniform superposition state to the desired target state by a process of induction.
Given the set of index states $x\in T_\ell$ that already satisfy the first $\ell$ partial oracle conditions $f_0(x)=0,\ldots,f_{\ell-1}(x)=0$ we define the following superposition state:

\begin{equation}
\ket{T_\ell} = \frac{1}{\sqrt{|T_\ell|}}\,\sum_{x\in T_\ell}\,\ket{x}
\end{equation}

Where we note that $|T_\ell|=2^{n-\ell}$.

We now show how, by applying the $\ell+1$ partial oracle iteration, we get the next state $T_{\ell+1}$ in the sequence, which satisfies $\ell+1$ partial oracle conditions: $f_0(x)=0,\ldots,f_\ell(x)=0$.

Define the indices $\{y_j\}_0^{n-1}$ such that $y_0=f_0(x),\ldots,y_{n-1}=f_{n-1}(x)$. Recalling that $f(x)$ is an explicitly bijective function (and thus invertible), we can write a given index $x$ value in terms of the $\{y_j\}$ indices as: $x=f^{-1}(y_{n-1}\ldots y_0)$.
Moreover, for all $x\in T_\ell$, we have $y_0=0,\ldots,y_{\ell-1}=0$, because all $x\in T_\ell$ satisfy the first $\ell$ partial oracle conditions.
Thus we can write the initial state as:

\begin{equation}
\ket{T_\ell} = \frac{1}{\sqrt{2^{n-\ell}}}\,\sum_{y_\ell,\ldots,y_{n-1}}\,\ket{f^{-1}(y_{n-1}\ldots y_\ell\, 0\ldots 0)}
\end{equation}

Now we apply the first operator from the partial oracle iteration:

\begin{equation}
e^{+i\frac{\pi}{2}\,f_\ell(x)}\,\ket{T_\ell}
  = \frac{1}{\sqrt{2^{n-\ell}}}\,
  \sum_{y_\ell,\ldots,y_{n-1}}\,e^{+i\frac{\pi}{2}\,y_\ell}
  \ket{f^{-1}(y_{n-1}\ldots y_\ell\, 0\ldots 0)}
\end{equation}

Now we apply the operator combination $R[f(x)]\,H^{\otimes n}$---noting from \ref{eq:baretransform} that we have $R[f(x)]\,H^{\otimes n}=B[f(x)]\,H^{\otimes n}\,H^{\otimes n}=B[f(x)]$ (as the Hadamard operator $H$ is its own inverse)---to give:

\begin{align*}
R[f(x)] & \,H^{\otimes n}\,e^{+i\frac{\pi}{2}\,f_\ell(x)}\,\ket{T_\ell} \\
  &= \left(
    \frac{1}{\sqrt{2^n}}
    \sum_\kappa \sum_x \, e^{-i2\pi\,\kappa\cdot f(x)}
    \ket{\kappa}\bra{x}
  \right)
  \left(
    \frac{1}{\sqrt{2^{n-\ell}}}\,
    \sum_{y_\ell,\ldots,y_{n-1}}\,e^{+i\frac{\pi}{2}\,y_\ell}
    \ket{f^{-1}(y_{n-1}\ldots y_\ell\, 0\ldots 0)}
  \right) \\
  &= 
  \frac{1}{\sqrt{2^n}}
  \frac{1}{\sqrt{2^{n-\ell}}}\,
  \sum_\kappa \sum_{y_\ell,\ldots,y_{n-1}} \, e^{-i2\pi\,\kappa\cdot f(x)}
  \ket{\kappa}\bra{f^{-1}(y_{n-1}\ldots y_\ell\, 0\ldots 0)}
  \,e^{+i\frac{\pi}{2}\,y_\ell}
  \ket{f^{-1}(y_{n-1}\ldots y_\ell\, 0\ldots 0)} \\
  &= 
  \frac{1}{\sqrt{2^n}}
  \frac{1}{\sqrt{2^{n-\ell}}}\,
  \sum_\kappa \sum_{y_\ell,\ldots,y_{n-1}}\,
  e^{-i\pi\,\sum_{j=\ell}^{n-1}\kappa_j y_j} \,e^{+i\frac{\pi}{2}\,y_\ell} \,\ket{\kappa}
\end{align*}

Given that $\ket{\kappa}=\ket{\kappa_{n-1}}\otimes\ket{\kappa_{n-2}}\otimes\cdots\otimes\ket{\kappa_0}$, we can factor this last expression into a product of $\kappa_j$ qubit states, as follows:

\begin{itemize}
\item For all of the $\kappa_j$ qubits where $j < \ell$, there is no sum over $y_j$, because $y_j=0$ for $j < \ell$. This also means there is no corresponding exponential coefficient $e^{-i\pi\,\kappa_j y_j}$ for the $\kappa_j$ qubit, so when we sum over $\kappa_j$ we get $(\ket{\kappa_j=0} + \ket{\kappa_j=1})$. Putting together all of the qubits for $j < \ell$ gives: $\left(\ket{\kappa_{\ell-1}=0} + \ket{\kappa_{\ell-1}=1}\right)\otimes\cdots\otimes\left(\ket{\kappa_0=0} + \ket{\kappa_0=1}\right)$.

\item The $l^{th}$ qubit $\kappa_\ell$ is a special case, because this qubit has the extra phase term $\exp(+i\frac{\pi}{2}y_\ell)$.
This qubit term can be evaluated as follows:

\begin{align*}
\sum_{\kappa_\ell=0}^1 \sum_{y_\ell=0}^1 \, e^{-i\pi\,\kappa_\ell y_\ell} e^{+i\frac{\pi}{2}\,y_\ell}\, \ket{\kappa_\ell}
  &= (1+i)\,\ket{\kappa_\ell=0} + (1-i)\,\ket{\kappa_\ell=1} \\
  &= 2\, e^{i\frac{\pi}{4}}\,\,\frac{1}{\sqrt{2}} \big( \ket{\kappa_\ell=0} - i \ket{\kappa_\ell=1} \big)
\end{align*}

\item For all of the $\kappa_j$ qubits where $\ell < j$, we can evaluate each of these qubit terms as follows:

\begin{align*}
\sum_{\kappa_\ell=0}^1 \sum_{y_\ell=0}^1 \, e^{-i\pi\,\kappa_\ell y_\ell} \, \ket{\kappa_\ell}
  &= 2\,\sum_{\kappa_\ell=0}^1 \, \delta_{\kappa_\ell, 0} \, \ket{\kappa_\ell} \\
  &= 2\, \ket{\kappa_\ell=0}
\end{align*}
\end{itemize}

Putting these terms together, we get the following expression:

\begin{align*}
R[f(x)] \,H^{\otimes n} & \,e^{+i\frac{\pi}{2}\,f_\ell(x)}\,\ket{T_\ell} \\
  &= e^{i\frac{\pi}{4}}\,
     \ket{\kappa_{n-1}=0}\otimes\cdots\otimes\ket{\kappa_{\ell+1}=0}
     \otimes\frac{1}{\sqrt{2}}\big(\ket{\kappa_\ell=0} - i\ket{\kappa_\ell=1} \big) \\
  & \>\>\>\> \>\>\>\> \otimes\frac{1}{\sqrt{2}}\big(\ket{\kappa_{\ell-1}=0} + \ket{\kappa_{\ell-1}=1} \big) \otimes\cdots
    \otimes\frac{1}{\sqrt{2}}\big(\ket{\kappa_0=0} + \ket{\kappa_0=1} \big)
\end{align*}

Now we apply the operator $e^{+i\frac{\pi}{2}\,\kappa_\ell}$, to get:

\begin{align*}
e^{+i\frac{\pi}{2}\,\kappa_\ell}\,R[f(x)] & \,H^{\otimes n} \,e^{+i\frac{\pi}{2}\,f_\ell(x)}\,\ket{T_\ell} \\
  &= e^{i\frac{\pi}{4}}\,
     \ket{\kappa_{n-1}=0}\otimes\cdots\otimes\ket{\kappa_{\ell+1}=0}
     \otimes\frac{1}{\sqrt{2}}\big(\ket{\kappa_\ell=0} + \ket{\kappa_\ell=1} \big) \\
  & \>\>\>\> \>\>\>\> \otimes\frac{1}{\sqrt{2}}\big(\ket{\kappa_{\ell-1}=0} + \ket{\kappa_{\ell-1}=1} \big) \otimes\cdots
    \otimes\frac{1}{\sqrt{2}}\big(\ket{\kappa_0=0} + \ket{\kappa_0=1} \big)
\end{align*}

Finally, applying the adjoint operator $H^{\otimes n}R^\dagger[f(x)]$ (which is equal to $B^\dagger[f(x)]$) to this state and summing over the $\kappa_j$ indices, we obtain:

\begin{align*}
H^{\otimes n}R^\dagger[f(x)]\,e^{+i\frac{\pi}{2}\,\kappa_\ell}\,R[f(x)] & \,H^{\otimes n} \,e^{+i\frac{\pi}{2}\,f_\ell(x)}\,\ket{T_\ell} \\
  &= \frac{e^{i\frac{\pi}{4}}}{\sqrt{2^{n-(\ell+1)}}}
  \sum_{y_0\ldots y_{n-1}}\,
  \delta_{y_\ell, 0}\ldots\delta_{y_0, 0}\,
  \ket{f^{-1}(y_{n-1}\ldots y_0)} \\
  &= \frac{e^{i\frac{\pi}{4}}}{\sqrt{2^{n-(\ell+1)}}}
  \sum_{y_{l+1}\ldots y_{n-1}}\,\ket{f^{-1}(y_{n-1}\ldots y_{\ell+1}\,0\ldots 0)} \\
  &= e^{i\frac{\pi}{4}}\,\ket{T_{\ell+1}}
\end{align*}

Where $\ket{T_{\ell+1}}$ is the state that satisfies the first $\ell+1$ partial oracle conditions: $f_0(x)=0,\ldots,f_\ell(x)=0$.

\subsection{Parallelization}
\label{parallelization}

Although it seems natural to apply the partial oracle conditions $\{f_\ell(x)\}_0^{n-1}$ sequentially---applying the partial oracle iteration of equation \ref{eq:partialoracleiteration-sequential} for each $\ell$---when we review the proof of the partial oracles algorithm in section \ref{partialoraclealgorithm}, we notice that each of the $\kappa_\ell$ qubits in reciprocal space get resolved independently of the others.
This suggests that it might be possible to apply all of the partial oracle conditions in a \textit{single} iteration (parallelization), instead of iterating $n$ times.
It turns out that this is indeed the case.

Equation \ref{eq:partialoracleiteration-single} shows the \textit{parallelized partial oracle iteration}, which can be applied to a uniform superposition state to obtain the desired target state in just a single iteration:

\begin{equation}
H^{\otimes n}\,R^\dagger[f(x)]\,e^{+i\frac{\pi}{2}\,\sum_{\ell=0}^{n-1}\kappa_\ell}\,R[f(x)]\,H^{\otimes n}\,e^{+i\frac{\pi}{2}\,\sum_{\ell=0}^{n-1} f_\ell(x)}
\label{eq:partialoracleiteration-single}
\end{equation}

Where the operator $e^{+i\frac{\pi}{2}\,\sum_{\ell=0}^{n-1} f_\ell(x)}$ applies the phase factor $e^{+i\frac{\pi}{2}}$ for every partial oracle condition, and in reciprocal space, the operator $e^{+i\frac{\pi}{2}\,\sum_{\ell=0}^{n-1}\kappa_\ell}$ applies the phase factor $e^{+i\frac{\pi}{2}}$ to all of the corresponding reciprocal qubits $\{\kappa_\ell\}_0^{n-1}$.

We recap some of the steps from the preceding proof of section \ref{partialoraclealgorithm} to see how the parallelized iteration works.
We start with the uniform superposition state:

\begin{equation*}
\ket{S} = \frac{1}{\sqrt{2^n}}\,\sum_x\,\ket{x}
\end{equation*}

After applying the operations $R[f(x)]\,H^{\otimes n}\,e^{+i\frac{\pi}{2}\,\sum_{\ell=0}^{n-1} f_\ell(x)}$ to this initial state, we get:

\begin{align*}
R[f(x)] \,H^{\otimes n} & \,e^{+i\frac{\pi}{2}\,\sum_{\ell=0}^{n-1} f_\ell(x)}\,\ket{S} \\
  &= e^{i\frac{\pi}{4}n}\,
     \frac{1}{\sqrt{2}}\big(\ket{\kappa_{n-1}=0} - i\ket{\kappa_{n-1}=1} \big)
     \otimes\cdots\otimes
     \frac{1}{\sqrt{2}}\big(\ket{\kappa_0=0} - i\ket{\kappa_0=1} \big)
\end{align*}

Now we apply the operator $e^{+i\frac{\pi}{2}\,\sum_{\ell=0}^{n-1}\kappa_\ell}$, to get:

\begin{align*}
e^{+i\frac{\pi}{2}\,\sum_{\ell=0}^{n-1}\kappa_\ell}\,R[f(x)] & \,H^{\otimes n} \,e^{+i\frac{\pi}{2}\,\sum_{\ell=0}^{n-1} f_\ell(x)}\,\ket{S} \\
  &= e^{i\frac{\pi}{4}n}\,
     \frac{1}{\sqrt{2}}\big(\ket{\kappa_{n-1}=0} + \ket{\kappa_{n-1}=1} \big)
     \otimes\cdots\otimes
     \frac{1}{\sqrt{2}}\big(\ket{\kappa_0=0} + \ket{\kappa_0=1} \big)
\end{align*}

Finally, applying the adjoint operator $H^{\otimes n}R^\dagger[f(x)]$ to this state and summing over the $\kappa_j$ indices, we obtain:

\begin{align*}
H^{\otimes n}R^\dagger[f(x)]\,e^{+i\frac{\pi}{2}\,\sum_{\ell=0}^{n-1}\kappa_\ell}
\,R[f(x)] & \,H^{\otimes n} \,e^{+i\frac{\pi}{2}\,\sum_{\ell=0}^{n-1} f_\ell(x)}\,\ket{S} \\
  &= e^{i\frac{\pi}{4}n}
  \sum_{y_0\ldots y_{n-1}}\,
  \delta_{y_{n-1}, 0}\ldots\delta_{y_0, 0}\,
  \ket{f^{-1}(y_{n-1}\ldots y_0)} \\
  &= e^{i\frac{\pi}{4}n} \ket{f^{-1}(0\ldots 0)} \\
  &= e^{i\frac{\pi}{4}n}\,\ket{T_n}
\end{align*}

Where $\ket{T_n}$ is the final target state that satisfies all $n$ of the partial oracle conditions.

Note that if you prefer to use the \textit{match on all ones} convention for the oracle (see section \ref{onetooneoracle}), by inspection of the preceding proof, we see that by replacing $e^{+i\frac{\pi}{2}\,\sum_{\ell=0}^{n-1}\kappa_\ell}$ by $e^{-i\frac{\pi}{2}\,\sum_{\ell=0}^{n-1}\kappa_\ell}$ (flipping the sign in the exponential) in equation \ref{eq:partialoracleiteration-single}, the parallelized partial oracle iteration returns the state $\ket{f^{-1}(1\ldots 1)}$ in the final step.

\subsection{Algorithm summary}
\label{algorithmsummary}

We can summarize the steps for implementing the parallelized partial oracles algorithm, as follows:

\begin{enumerate}
\item Before you start, calculate the reciprocal transform $R[f(x)]$ of the oracle function and figure out the quantum circuits needed to implement this operator (examples are given in section \ref{applicationtosha}).

\item Initialize the index register $x\in\{0,1\}^n$ in a uniform superposition state, by applying the Walsh-Hadamard transformation:
\[ H^{\otimes n}\,\ket{0} \]

\item Calculate the value of the multi-qubit oracle function $\{f_j(x)\}_0^{n-1}$ (which typically requires using ancillary qubits, in addition to the index qubits).

\item In order to apply the $e^{+i\frac{\pi}{2}\,\sum_{\ell=0}^{n-1} f_\ell(x)}$ operator, execute the phase gate (S gate) on all of the oracle qubits $\{f_j(x)\}_0^{n-1}$ (that is, the qubits that contain the result of calculating the multi-qubit oracle).

\item Uncompute the multi-qubit oracle function $\{f_j(x)\}_0^{n-1}$.

\item Map the index qubits to reciprocal space, by applying a Walsh-Hadamard transformation $H^{\otimes n}$.

\item Apply the circuit that implements the reciprocal operator $R[f(x)]$ (using the circuit worked out in step 1).

\item In order to apply the $e^{+i\frac{\pi}{2}\,\sum_{\ell=0}^{n-1}\kappa_\ell}$ operator, execute the phase gate (S gate) on all of the reciprocal space qubits $\{\kappa_j\}_0^{n-1}$.

\item Uncompute the reciprocal operator, by applying its adjoint operator $R^\dagger[f(x)]$.

\item Map the index qubits from reciprocal space back to direct space, by applying a Walsh-Hadamard transformation $H^{\otimes n}$ to the index qubits.

\item Measure the qubits from the index register to obtain the solution $x_s$ that satisfies the oracle constraint $f(x_s)=0$.

\end{enumerate}

\subsection{Chain rule theorem}
\label{chainruletheorem}

The reciprocal transform obeys a chain rule, which provides a powerful tool for decomposing complex transforms.

\begin{theorem}[Chain rule]
Given the bijective nested function $f^\sigma(g^\mu(x^\nu))$, the reciprocal transform of this nested function obeys the following identity:

\begin{equation}
\label{eq:chainruletheorem}
R[f^\sigma (g^\mu(x^\nu))] = R[f^\sigma (y^\mu)] \, R[g^\mu(x^\nu)]
\end{equation}
\end{theorem}

\begin{proof}
Consider the product of reciprocal transforms:

\begin{align*}
R[f^\sigma (y^\mu)] \, R[g^\mu(x^\nu)] \\
  = \frac{1}{\sqrt{N_\sigma N_\nu}}\frac{1}{N_\mu} &
    \sum_{\kappa^{\prime\sigma}}\sum_{y^\mu}\sum_{k^{\prime\mu}}
    \sum_{\kappa^\mu}\sum_{x^\nu}\sum_{k^\nu}
    e^{i2\pi\,(-\kappa^{\prime\sigma}\cdot f^\sigma(y^\mu) + y^\mu\cdot k^{\prime\mu} -\kappa^\mu\cdot g^\mu(x^\nu) + x^\nu\cdot k^\nu )}
    |\kappa^{\prime\sigma}\rangle\langle k^{\prime\mu}|\kappa^\mu\rangle\langle k^\nu|
    \\
  = \frac{1}{\sqrt{N_\sigma N_\nu}}\frac{1}{N_\mu} &
    \sum_{\kappa^{\prime\sigma}}\sum_{y^\mu}
    \sum_{\kappa^\mu}\sum_{x^\nu}\sum_{k^\nu}
    e^{i2\pi\,(-\kappa^{\prime\sigma}\cdot f^\sigma(y^\mu) + \kappa^\mu\cdot(y^\mu- g^\mu(x^\nu)) + x^\nu\cdot k^\nu )}
    |\kappa^{\prime\sigma}\rangle\langle k^\nu|
    \\
  = \frac{1}{\sqrt{N_\sigma N_\nu}}
    \sum_{\kappa^{\prime\sigma}} & \sum_{y^\mu}
    \sum_{x^\nu}\sum_{k^\nu}
    e^{i2\pi\,(-\kappa^{\prime\sigma}\cdot f^\sigma(y^\mu) + x^\nu\cdot k^\nu )}
    \delta\left(y^\mu - g^\mu(x^\nu)\right)
    |\kappa^{\prime\sigma}\rangle\langle k^\nu|
    \\
  = \frac{1}{\sqrt{N_\sigma N_\nu}}
    \sum_{\kappa^{\prime\sigma}} &
    \sum_{x^\nu}\sum_{k^\nu}
    e^{i2\pi\,(-\kappa^{\prime\sigma}\cdot f^\sigma(g^\mu(x^\nu)) + x^\nu\cdot k^\nu )}
    |\kappa^{\prime\sigma}\rangle\langle k^\nu|
    \>\> = \>\> R[f^\sigma (g^\mu(x^\nu))]
\end{align*}

Where, in the first step, we use the delta function $\langle k^{\prime\mu}|\kappa^\mu\rangle$ to eliminate the variable $k^{\prime\mu}$ and replace it by $\kappa^\mu$ and, in the second step, we use equation \ref{eq:deltafunction} to sum over $\kappa^\mu$ and obtain the delta function $\delta\left(y^\mu - g^\mu(x^\nu)\right)$.

\end{proof}

\subsection{Target value not required in reciprocal operators}

In a typical search problem, we are given a one-to-one function $g^\sigma(x^\mu):\{0,1\}^n\rightarrow\{0,1\}^n$ and we then search for the value of $x^\mu$ for which $g^\sigma(x^\mu)$ returns a specific target value $y^\sigma_t$.  In other words, we search for the value (or values) of $x^\mu$ that satisfy:

\begin{equation}
g^\sigma(x^\mu) = y^\sigma_t
\end{equation}

If we adopt the \textit{match on all zeros} convention, we can then define the corresponding oracle function for this problem as:

\begin{equation}
f^\sigma(x^\mu) = g^\sigma(x^\mu)\oplus y^\sigma_t
\end{equation}

Where $\oplus$ represents the bitwise XOR operation (that is, each bit $g^\sigma_j(x^\mu)$ is XORed with each bit $y^\sigma_{t, j}$).

Now, when it comes to calculating the reciprocal transform of $f^\sigma(x^\mu)$, we can substitute $f^\sigma(x^\mu) = g^\sigma(x^\mu)\oplus y^\sigma_t$ into the reciprocal transform equation \ref{eq:reciprocaltransform}. It might appear necessary to keep track of the value $y^\sigma_t$ while we are calculating the reciprocal transform, but this is not the case.

Consider the reciprocal part of the partial oracles algorithm:

\begin{align*}
R^\dagger & [f^\sigma(x^{\prime\mu})] \left( \prod_{j=0}^n e^{+i\pi\,\kappa_j} \right) R[f^\sigma(x^\mu)] \\
 &=
    \frac{1}{N_\sigma}\frac{1}{N_\mu}
    \sum_{k^{\prime\mu}}\sum_{x^{\prime\mu}}\sum_{\kappa^\sigma}
    e^{i2\pi\,(+\kappa^\sigma\cdot (g^\sigma(x^{\prime\mu})\oplus y^\sigma_t) - x^{\prime\mu}\cdot k^{\prime\mu} )}
    |k^{\prime\mu}\rangle\langle \kappa^\sigma|\kappa^\sigma\rangle
    \left( \prod_{j=0}^n e^{+i\pi\,\kappa_j} \right) \langle \kappa^\sigma| \\
 & \>\>\>\>\>\>\>\>\>\>\>\> \>\>\>\>\>\>\>\>\>\>\>\>
    \sum_{x^\mu}\sum_{k^\mu}
    e^{i2\pi\,(-\kappa^\sigma\cdot (g^\sigma(x^\mu)\oplus y^\sigma_t) + x^\mu\cdot k^\mu )}
    |\kappa^\sigma\rangle\langle k^\mu|
\end{align*}

When we compare the target term $+i2\pi\,\kappa^\sigma\cdot y^\sigma_t$ from the adjoint reciprocal operator $R^\dagger$ with the target term $-i2\pi\,\kappa^\sigma\cdot y^\sigma_t$ from the reciprocal operator $R$, we note that they have opposite signs and thus cancel each other out. Hence, the target value $y^\sigma_t$ has no significance here and we can conveniently set it to zero.

\section{Application to SHA algorithms}
\label{applicationtosha}

To illustrate how to apply the partial oracles algorithm in practice, we consider the problem of how to invert the SHA algorithms from the Secure Hash Standard \cite{SHS_FIPS_180_4}.

\subsection{SHA-256 functions and modulo addition}
\label{sha256functions}

The SHA-256 hash algorithm requires the following logical functions, which operate on the 32-bit words $x$, $y$, and $z$:

\begin{align}
Maj(x, y, z) &= (x\wedge y)\oplus(y\wedge z)\oplus(z\wedge x) \\
Ch(x, y, z)  &= (x\wedge y)\oplus(\neg x\wedge z) \\
\Sigma^{\{256\}}_0(x) &= ROTR^2(x) \oplus ROTR^{13}(x) \oplus ROTR^{22}(x) \\
\Sigma^{\{256\}}_1(x) &= ROTR^6(x) \oplus ROTR^{11}(x) \oplus ROTR^{25}(x) \\
\sigma^{\{256\}}_0(x) &= ROTR^7(x) \oplus ROTR^{18}(x) \oplus SHR^3(x) \\
\sigma^{\{256\}}_1(x) &= ROTR^{17}(x) \oplus ROTR^{19}(x) \oplus SHR^{10}(x)
\end{align}

The $ROTR^n(x)$ and $SHR^n(x)$ functions are defined as follows:

\begin{itemize}
\item $ROTR^n(x)$---given the $w$-bit word $x$ and integer value $n$ (where $0\le n < w$), the rotate right function is defined as $ROTR^n(x)=(x >> n)\vee (x << w - n)$.

\item $SHR^n(x)$---given the $w$-bit word $x$ and integer value $n$, the shift right function is defined as $SHR^n(x)=x >> n$.
\end{itemize}

The SHA-256 hash algorithm also makes use of addition modulo $2^{32}$. That is, wherever a $+$ sign occurs in the algorithm description, it is to be interpreted as addition modulo $2^{32}$ (equivalent to 32-bit addition with no carry from the final bit).

\subsection{Majority function Maj(x, y, z)}
\label{majorityfunction}

Consider the 1-bit word version of the majority function $Maj(a, b, c)$ (that is, where $a$, $b$, and $c$ each represent single qubits):

\begin{equation}
Maj(a, b, c) = (a\wedge b) \oplus (b\wedge c) \oplus (c\wedge a)
\end{equation}

We can then define the following oracle function (completing the oracle to make a 3-bit one-to-one function, as discussed in section \ref{onetooneoracle}):

\begin{align}
f_0 &= (a\wedge b) \oplus (b\wedge c) \oplus (c\wedge a) \\
f_1 &= a\oplus b \\
f_2 &= a\oplus c
\end{align}

Figure \ref{fig:majcircuit} shows the corresponding circuit for the majority oracle function, which requires just three gates to implement (for example, see the definition of the MAJ operation in the Cuccaro adder \cite{cuccaro2004newquantumripplecarryaddition}).

\begin{figure}
\begin{tikzpicture}
    \begin{yquant}
        qubit {$a$} a;
        qubit {$b$} b;
        qubit {$c$} c;
        hspace {4mm} -;

		cnot b | a;
		cnot c | a;
		cnot a | b, c;
		
        hspace {4mm} -;
        output {$f_0$} a;
        output {$f_1$} b;
        output {$f_2$} c;
    \end{yquant}
\end{tikzpicture}
\caption{Circuit for the majority function $Maj(a,b,c)$}\label{fig:majcircuit}
\end{figure}
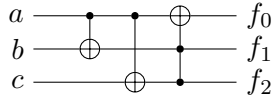

For calculating the reciprocal transform of the majority function, it is useful to know the inverse of the oracle function, which is given by:

\begin{align}
a &= f_0\oplus(f_1\wedge f_2) \label{eq:majinverse_a} \\
b &= f_1\oplus f_0\oplus(f_1\wedge f_2) \label{eq:majinverse_b} \\
c &= f_2\oplus f_0\oplus(f_1\wedge f_2) \label{eq:majinverse_c}
\end{align}

From equation \ref{eq:reciprocaltransform}, the reciprocal transform of the majority gate is given by:

\begin{align*}
R[f(a,b,c)] = \frac{1}{2^3}
    \sum_{\kappa}\sum_{a, b, c}\sum_{k}
    e^{i\pi\,(-(\kappa_0 f_0 + \kappa_1 f_1 + \kappa_2 f_2) + a k_a + b k_b + c k_c )}
    \ket{\kappa_0}\ket{\kappa_1}\ket{\kappa_2}\bra{k_a}\bra{k_b}\bra{k_c}
\end{align*}

Using equations \ref{eq:majinverse_a} to \ref{eq:majinverse_c} to replace $a$, $b$, and $c$ in the preceding expression, we get:

\begin{align*}
\frac{1}{2^3}
\sum_{\kappa}\sum_{f_0, f_1, f_2}\sum_{k} &
e^{i\pi\,(-(\kappa_0 f_0 + \kappa_1 f_1 + \kappa_2 f_2) + (f_0 + f_1 f_2) k_a + (f_1 + f_0 + f_1 f_2) k_b + (f_2 + f_0 + f_1 f_2) k_c )} \\
& \>\>\>\>\>\>\>\> \>\>\>\>\>\>\>\> \times \ket{\kappa_0}\ket{\kappa_1}\ket{\kappa_2}\bra{k_a}\bra{k_b}\bra{k_c}
\end{align*}

By summing over $f_0$ (taking the values $0$ and $1$), then $f_1$, and then $f_2$, we get the expression:

\begin{equation}
\frac{1}{2}\delta_{\kappa_0, p}\left( \delta_{\kappa_1, k_b} + \delta_{\kappa_1, p\oplus k_b} (-1)^{k_c\oplus\kappa_2} \right)\,
\ket{\kappa_0}\ket{\kappa_1}\ket{\kappa_2}\bra{k_a}\bra{k_b}\bra{k_c}
\label{eq:recipmajoritymatrix}
\end{equation}

Where $p$ is the parity $p=k_a\oplus k_b\oplus k_c$. If we consider the case where $p=0$, then the expression reduces to $\delta_{\kappa_0,0}\delta_{\kappa_1,k_b}\delta_{\kappa_2,k_c}$, so that qubits $k_b$ and $k_c$ pass straight through, unchanged. On the other hand, if we consider the case where $p=1$, the expression becomes $\frac{1}{2}\delta_{\kappa_0, 1}\left( \delta_{\kappa_1, k_b} + \delta_{\kappa_1, \overline{k_b}} (-1)^{k_c\oplus\kappa_2} \right)$, where $\overline{k_b}$ is the logical negation of $k_b$.

The term in parentheses $\delta_{\kappa_1, k_b} + \delta_{\kappa_1, \overline{k_b}} (-1)^{k_c\oplus\kappa_2}$ represents a 4x4 Walsh matrix, which is an orthogonal (therefore, also unitary) matrix whose elements are either $+1$ or $-1$. But this matrix is not the familiar 4x4 matrix we get by applying $H^{\otimes 2}$, as it has only one minus sign in each row or column. With a little trial and error, we find that this Walsh matrix can be implemented by the circuit shown in figure \ref{fig:walshcircuit}. This circuit requires on ancillary qubit initialized in the state $\ket{-}=\left(\ket{0} - \ket{1}\right)/\sqrt{2}$ and the phase kickback from $\ket{-}$ is used to flip the signs on the last row and the last column of the matrix.

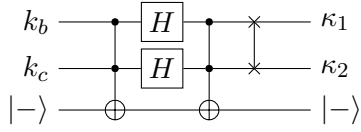
\begin{figure}

\begin{tikzpicture}
    \begin{yquant}
        qubit {$k_b$} kb;
        qubit {$k_c$} kc;
        qubit {$\ket{-}$} ph;
        hspace {4mm} -;

		cnot ph | kb, kc;
		h kb;
		h kc;
		cnot ph | kb, kc;
		swap (kb, kc);
		
        hspace {4mm} -;
        output {$\kappa_1$} kb;
        output {$\kappa_2$} kc;
        output {$\ket{-}$} ph;
    \end{yquant}
\end{tikzpicture}

\caption{Circuit for the $\delta_{\kappa_1, k_b} + \delta_{\kappa_1, \overline{k_b}} (-1)^{k_c\oplus\kappa_2}$ Walsh matrix}\label{fig:walshcircuit}
\end{figure}

To get the complete circuit for the expression in equation \ref{eq:recipmajoritymatrix}, we note that the $k_a$ qubit needs to be initialized with the parity $p=k_a\oplus k_b\oplus k_c$ and the execution of the 4x4 Walsh matrix needs to be made conditional on this parity. This gives us the circuit shown in figure \ref{fig:recipmajoritygate}, which is the complete circuit for the reciprocal majority gate.

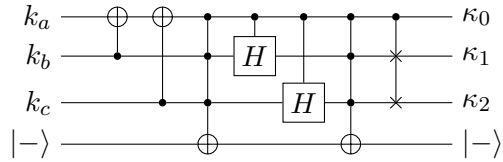
\begin{figure}[h]
\centering
\begin{tikzpicture}
    \begin{yquant}
        qubit {$k_a$} ka;
        qubit {$k_b$} kb;
        qubit {$k_c$} kc;
        qubit {$\ket{-}$} ph;
        hspace {4mm} -;

		cnot ka | kb;
		cnot ka | kc;
		cnot ph | ka, kb, kc;
		h kb | ka;
		h kc | ka;
		cnot ph | ka, kb, kc;
		swap (kb, kc) | ka;
		
        hspace {4mm} -;
        output {$\kappa_0$} ka;
        output {$\kappa_1$} kb;
        output {$\kappa_2$} kc;
        output {$\ket{-}$} ph;
    \end{yquant}
\end{tikzpicture}

\caption{Reciprocal majority gate}\label{fig:recipmajoritygate}
\end{figure}

\subsection{Choose function Ch(x, y, z)}
\label{choosefunction}

Consider the 1-bit word version of the choose function $Ch(a,b,c)$ (that is, where $a$, $b$, and $c$ each represent single qubits):

\begin{equation}
Ch(a, b, c)  = (a\wedge b)\oplus(\neg a\wedge c)
\end{equation}

We can then define the following complete 3-bit oracle function $f(a,b,c):\{0,1\}^3\rightarrow\{0,1\}^3$:

\begin{align}
f_0 &= a \\
f_1 &= b\oplus c \\
f_2 &= (a\wedge b)\oplus(\neg a\wedge c)
\end{align}

Figure \ref{fig:chcircuit} shows the corresponding circuit for the choose function, which requires just two gates to implement.

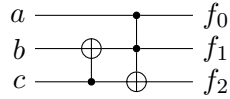
\begin{figure}[h]
\centering
\begin{tikzpicture}
    \begin{yquant}
        qubit {$a$} a;
        qubit {$b$} b;
        qubit {$c$} c;
        hspace {4mm} -;

		cnot b | c;
		cnot c | a, b;
		
        hspace {4mm} -;
        output {$f_0$} a;
        output {$f_1$} b;
        output {$f_2$} c;
    \end{yquant}
\end{tikzpicture}

\caption{Circuit for the choose function $Ch(a,b,c)$}\label{fig:chcircuit}
\end{figure}

The inverse of the choose oracle function is given by:

\begin{align}
a &= f_0 \label{eq:chinverse_a} \\
b &= f_2 + (1-f_0)f_1 \label{eq:chinverse_b} \\
c &= f_2 + f_0 f_1 \label{eq:chinverse_c}
\end{align}

From equation \ref{eq:reciprocaltransform}, the reciprocal transform of the choose gate is given by:

\begin{align*}
R[f(a,b,c)] = \frac{1}{2^3}
    \sum_{\kappa}\sum_{a, b, c}\sum_{k}
    e^{i\pi\,(-(\kappa_0 f_0 + \kappa_1 f_1 + \kappa_2 f_2) + a k_a + b k_b + c k_c )}
    \ket{\kappa_0}\ket{\kappa_1}\ket{\kappa_2}\bra{k_a}\bra{k_b}\bra{k_c}
\end{align*}

Using equations \ref{eq:chinverse_a} to \ref{eq:chinverse_c} to replace $a$, $b$, and $c$ in the preceding expression, we get:

\begin{align*}
\frac{1}{2^3}
\sum_{\kappa}\sum_{f_0, f_1, f_2}\sum_{k} &
e^{i\pi\,(-(\kappa_0 f_0 + \kappa_1 f_1 + \kappa_2 f_2) + f_0 k_a + (f_2 + (1-f_0)f_1) k_b + (f_2 + f_0 f_1) k_c )} \\
& \>\>\>\>\>\>\>\> \>\>\>\>\>\>\>\> \times \ket{\kappa_0}\ket{\kappa_1}\ket{\kappa_2}\bra{k_a}\bra{k_b}\bra{k_c}
\end{align*}

By summing over $f_0$ (taking the values $0$ and $1$), then $f_1$, and then $f_2$, we get the expression:

\begin{equation}
\frac{1}{2}\delta_{\kappa_2, p}\left( \delta_{\kappa_0, k_a} + \delta_{\kappa_0, p\oplus k_a} (-1)^{k_b\oplus\kappa_1} \right)\,
\ket{\kappa_0}\ket{\kappa_1}\ket{\kappa_2}\bra{k_a}\bra{k_b}\bra{k_c}
\label{eq:recipchoosematrix}
\end{equation}

Where $p=k_b\oplus k_c$. If we consider the case where $p=0$, the expression reduces to $\delta_{\kappa_2, 0}\delta_{\kappa_0, k_a}\delta_{\kappa_1, k_b}$, so that qubits $k_a$ and $k_b$ pass straight through, unchanged. On the other hand, if we consider the case where $p=1$, the expression becomes $\delta_{\kappa_2, 1}\left( \delta_{\kappa_0, k_a} + \delta_{\kappa_0, \overline{k_a}} (-1)^{k_b\oplus\kappa_1} \right)$.

Similarly to the case of the majority function, the term in parentheses $\delta_{\kappa_0, k_a} + \delta_{\kappa_0, \overline{k_a}} (-1)^{k_b\oplus\kappa_1}$ represents a 4x4 Walsh matrix.

To get the complete circuit for the expression in equation \ref{eq:recipchoosematrix}, we note that the $k_c$ qubit needs to be initialized with the parity $p=k_b\oplus k_c$ and the execution of the 4x4 Walsh matrix needs to be made conditional on this parity. This gives us the circuit shown in figure \ref{fig:recipchoosegate}, which is the complete circuit for the reciprocal choose gate.

\begin{figure}[h]
\centering
\begin{tikzpicture}
    \begin{yquant}
        qubit {$k_a$} ka;
        qubit {$k_b$} kb;
        qubit {$k_c$} kc;
        qubit {$\ket{-}$} ph;
        hspace {4mm} -;

		cnot kc | kb;
		cnot ph | ka, kb, kc;
		h ka | kc;
		h kb | kc;
		cnot ph | ka, kb, kc;
		swap (ka, kb) | kc;
		
        hspace {4mm} -;
        output {$\kappa_0$} ka;
        output {$\kappa_1$} kb;
        output {$\kappa_2$} kc;
        output {$\ket{-}$} ph;
    \end{yquant}
\end{tikzpicture}

\caption{Reciprocal choose gate}\label{fig:recipchoosegate}
\end{figure}

\subsection{Modulo addition}

The SHA algorithms use a simple kind of modulo addition that involves adding two $w$-bit words and discarding (or simply not calculating) the highest order carry bit. For example, in the case of 32-bit word modulo addition, this corresponds to addition modulo $2^{32}$.

To calculate the reciprocal transform of a modulo adder gate, we exploit the chain rule theorem from section \ref{chainruletheorem}: using the chain rule, the sequence of carry gates and sum gates that we use to define an adder gate can be converted into a corresponding sequence of \textit{reciprocal carry} and \textit{reciprocal sum} gates in the reciprocal adder gate.

Before we can calculate the reciprocal adder gate, however, we need an ordinary adder gate, which can serve as a template for the reciprocal circuit. There are two building blocks we need for the adder gate:

\begin{itemize}
\item \textit{Carry gate} ($C$)---it turns out that we already have a circuit that implements a carry gate. The $Maj(a,b,c)$ gate shown in figure \ref{fig:majcircuit} behaves as an \textit{in-place carry gate} \cite{cuccaro2004newquantumripplecarryaddition}. Before applying this carry gate, qubit $a$ corresponds to the carry-in bit and qubits $b$ and $c$ correspond to the bits from the summands. After applying this gate, qubit $a$ corresponds to the carry-out bit, while qubits $b$ and $c$ are left in an unclean state (subsequently fixed when we uncompute the carry gate). 

\item \textit{Sum gate} ($S$)---has the circuit implementation shown in figure \ref{fig:sumgate}. In this circuit, qubit $c$ corresponds to carry-in, while qubits $a$ and $b$ correspond to the bits from the summands. After applying this gate, qubit $b$ holds the sum $a\oplus b\oplus c$, while qubits $a$ and $c$ remain unchanged.
\end{itemize}

\begin{figure}[h]
\centering
\begin{tikzpicture}
    \begin{yquant}
        qubit {$c$} c;
        qubit {$a$} a;
        qubit {$b$} b;
        hspace {4mm} -;

		cnot b | a;
		cnot b | c;
		
        hspace {4mm} -;
        output {$c$} c;
        output {$a$} a;
        output {$s=a\oplus b\oplus c$} b;
    \end{yquant}
\end{tikzpicture}
\caption{Sum gate}\label{fig:sumgate}
\end{figure}
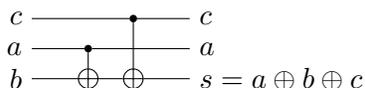

The circuit for the modulo adder gate (in direct space) is shown in figure \ref{fig:addergate}. This adder acts on two 4-qubit registers $x^{(1)}$ and $x^{(2)}$ and functions as an in-place modulo-$16$ adder, returning the sum in register $x^{(2)}$. In other words, it calculates $x^{(2)} = x^{(2)} + x^{(1)}$.
The wavy line between boxes in figure \ref{fig:addergate} is used to skip over qubits not affected by the relevant gate. For example, the second carry gate from the left operates on the qubits \texttt{carry}, $x^{(1)}_1$, $x^{(2)}_1$, but not on the qubits $x^{(1)}_0$, $x^{(2)}_0$.
The carry gates in figure \ref{fig:addergate} are wired up such that the carry-out is always written to the \texttt{carry} bit.

\begin{figure}[h]
\centering
\begin{tikzpicture}
    \begin{yquant}
        qubit {carry $\ket{0}$} c;
        qubit {$x^{(1)}_0$} x10;
        qubit {$x^{(2)}_0$} x20;
        qubit {$x^{(1)}_1$} x11;
        qubit {$x^{(2)}_1$} x21;
        qubit {$x^{(1)}_2$} x12;
        qubit {$x^{(2)}_2$} x22;
        qubit {$x^{(1)}_3$} x13;
        qubit {$x^{(2)}_3$} x23;
        hspace {4mm} -;

	    box {$C$} (x10, x20, c);
        hspace {4mm} -;
	    box {$C$} (x11, x21, c);
        hspace {4mm} -;
		box {$C$} (x12, x22, c);
        hspace {4mm} -;
		box {$S$} (x13, x23, c);
        hspace {4mm} -;
		box {$C^\dagger$} (x12, x22, c);
        hspace {4mm} -;
		box {$S$} (x12, x22, c);
        hspace {4mm} -;
		box {$C^\dagger$} (x11, x21, c);
        hspace {4mm} -;
		box {$S$} (x11, x21, c);
        hspace {4mm} -;
		box {$C^\dagger$} (x10, x20, c);
        hspace {4mm} -;
		box {$S$} (x10, x20, c);
		
        hspace {4mm} -;
    \end{yquant}
\end{tikzpicture}

\caption{Modulo adder gate}\label{fig:addergate}
\end{figure}
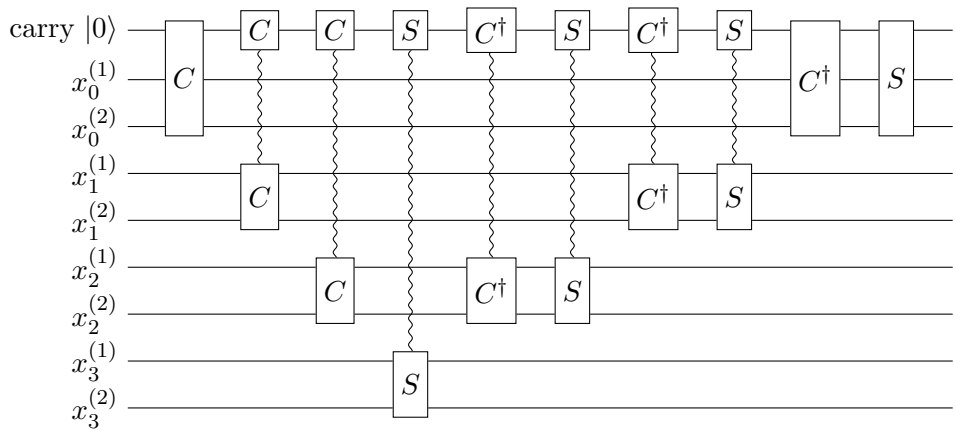

To implement the reciprocal modulo adder, we need the following building blocks in reciprocal space:

\begin{itemize}
\item \textit{Reciprocal carry gate} ($rC$)---since the $Maj(a,b,c)$ gate is equivalent to an in-place carry gate, it follows that the corresponding reciprocal majority gate from figure \ref{fig:recipmajoritygate} is equivalent to a reciprocal carry gate. When wiring up this gate, note that the $k_a/\kappa_0$ qubit from figure \ref{fig:recipmajoritygate} corresponds to the reciprocal carry qubit. Note also that this gate requires an additional $\ket{-}$ ancillary qubit to provide a phase kickback.

\item \textit{Reciprocal sum gate} ($rS$)---to obtain the reciprocal sum gate, we need to calculate the reciprocal of the circuit shown in figure \ref{fig:sumgate}, which we now proceed to do.
\end{itemize}

The oracle function corresponding to the sum gate from figure \ref{fig:sumgate} can be defined as:

\begin{align}
f_0 &= c \\
f_1 &= a \\
f_2 &= a\oplus b\oplus c
\end{align}

From equation \ref{eq:reciprocaltransform}, the reciprocal transform of the sum gate is given by:

\begin{equation}
\frac{1}{2^3}
\sum_{\kappa}\sum_{a,b,c}\sum_{k}
e^{i\pi\,(-(\kappa_0 c + \kappa_1 a + \kappa_2 (a\oplus b\oplus c)) + a k_a + b k_b + c k_c )}
\, \ket{\kappa_0}\ket{\kappa_1}\ket{\kappa_2}\bra{k_a}\bra{k_b}\bra{k_c}
\end{equation}

By summing over $a$ (taking the values $0$ and $1$), then $b$, and then $c$, we get the expression:

\begin{equation}
\delta_{k_a, \kappa_1\oplus\kappa_2}\delta_{k_b, \kappa_2}\delta_{k_c, \kappa_0\oplus\kappa_2}\,
\ket{\kappa_0}\ket{\kappa_1}\ket{\kappa_2}\bra{k_a}\bra{k_b}\bra{k_c}
\label{eq:recipsummatrix}
\end{equation}

Inverting the relation between $(k_a, k_b, k_c)$ and $(\kappa_a, \kappa_b, \kappa_c)$ gives us:

\begin{align*}
\kappa_0 &= k_b\oplus k_c \\
\kappa_1 &= k_b\oplus k_a \\
\kappa_2 &= k_b
\end{align*}

Which gives the reciprocal sum gate circuit, as shown in figure \ref{fig:recipsumgate}.

\begin{figure}[h]
\centering
\begin{tikzpicture}
    \begin{yquant}
        qubit {$k_c$} kc;
        qubit {$k_a$} ka;
        qubit {$k_b$} kb;
        hspace {4mm} -;

		cnot ka | kb;
		cnot kc | kb;
		
        hspace {4mm} -;
        output {$\kappa_0$} kc;
        output {$\kappa_1$} ka;
        output {$\kappa_2$} kb;
    \end{yquant}
\end{tikzpicture}
\caption{Reciprocal sum gate}\label{fig:recipsumgate}
\end{figure}
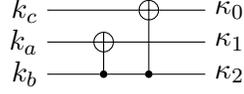

The circuit for the reciprocal modulo adder gate is shown in figure \ref{fig:recipaddergate}. The reciprocal adder acts on two 4-qubit registers $k^{(1)}$ and $k^{(2)}$ (which have been mapped by a Walsh-Hadamard transformation from the $x^{(1)}$ and $x^{(2)}$ registers in direct space). This circuit has a more or less analogous layout to the direct adder circuit from \ref{fig:addergate}, but with the reciprocal carry gate $rC$ in place of the carry gate and the reciprocal sum gate $rS$ in place of the sum gate.

Note that the reciprocal modulo adder gate has an additional ancillary qubit $\ket{-}$, which is needed by the reciprocal carry gate. Also note that, because the \texttt{carry} qubit is represented in reciprocal space, it must be initialized in the $\ket{+}$ state (not the $\ket{0}$ state, as used in direct space).

\begin{figure}[h]
\centering
\begin{tikzpicture}
    \begin{yquant}
        qubit {anc $\ket{-}$} ph;
        qubit {carry $\ket{+}$} c;
        qubit {$k^{(1)}_0$} k10;
        qubit {$k^{(2)}_0$} k20;
        qubit {$k^{(1)}_1$} k11;
        qubit {$k^{(2)}_1$} k21;
        qubit {$k^{(1)}_2$} k12;
        qubit {$k^{(2)}_2$} k22;
        qubit {$k^{(1)}_3$} k13;
        qubit {$k^{(2)}_3$} k23;
        hspace {4mm} -;

	    box {$rC$} (k10, k20, c, ph);
        hspace {4mm} -;
	    box {$rC$} (k11, k21, c, ph);
        hspace {4mm} -;
		box {$rC$} (k12, k22, c, ph);
        hspace {4mm} -;
		box {$rS$} (k13, k23, c);
        hspace {4mm} -;
		box {$rC^\dagger$} (k12, k22, c, ph);
        hspace {4mm} -;
		box {$rS$} (k12, k22, c);
        hspace {4mm} -;
		box {$rC^\dagger$} (k11, k21, c, ph);
        hspace {4mm} -;
		box {$rS$} (k11, k21, c);
        hspace {4mm} -;
		box {$rC^\dagger$} (k10, k20, c, ph);
        hspace {4mm} -;
		box {$rS$} (k10, k20, c);
		
        hspace {4mm} -;
    \end{yquant}
\end{tikzpicture}

\caption{Reciprocal modulo adder gate}\label{fig:recipaddergate}
\end{figure}
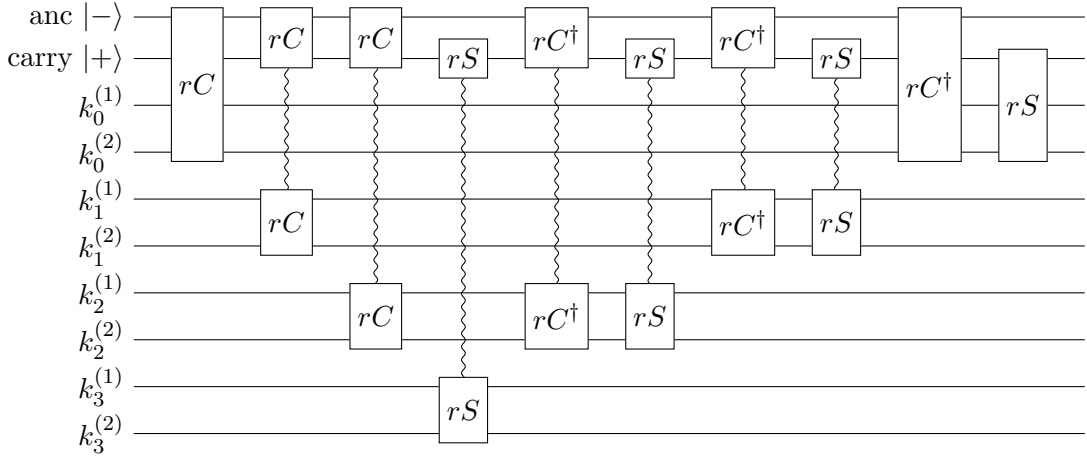

\subsection{Bit shifting functions}
\label{bitshiftingfunctions}

We have already introduced the standard bit-shifting functions needed for the SHA-256 hash algorithm in section \ref{sha256functions}.
In this section, we generalize the definitions, in order to make it easier to deal with a variety of different bit-shifting functions.

Given a bit width $w$, we can conveniently define define the basic bit-shifting functions $ROTR^a_j(x)$, $ROTL^b_j(x)$, $SHR^c_j(x)$, and $SHL^d_j(x)$, where the $\{j\}_0^{w-1}$ subscripts index the bits of the returned rotated/shifted register, as follows:

\begin{align}
ROTR^a_j(x) &= x_{j+a\textrm{ mod } w} & \forall a\in\mathbb{Z} 
\label{eq:rotrdefinition} \\
ROTL^b_j(x) &= x_{j-b\textrm{ mod } w} & \forall b\in\mathbb{Z} 
\label{eq:rotldefinition} \\
SHR^c_j(x)    &= x_{j+c}\,\theta(w-1-(j+c))\,\theta(j+c) & -w < c < w \textrm{ and } c\neq 0
\label{eq:shrdefinition} \\
SHL^d_j(x)    &= x_{j-d}\,\theta(w-1-(j-d))\,\theta(j-d) & -w < d < w \textrm{ and } d\neq 0
\label{eq:shldefinition}
\end{align}

The Heaviside theta function $\theta(x)$ is defined as: $\theta(x)=0$ for $x<0$ and $\theta(x)=1$ for $x\ge 0$.
Usually, the shift functions are defined in such a way that they are valid only for positive superscripts, but with the definitions given in equations \ref{eq:rotrdefinition} to \ref{eq:shldefinition}, these functions are valid for both positive and negative shifts.
This makes the right and left-shift functions more or less interchangeable, just by flipping the sign of the shift amount, for example:

\begin{align*}
ROTL^a(x) &= ROTR^{-a}(x) \\
SHL^c(x)  &= SHR^{-c}(x)
\end{align*}

We can now define the form of the general shift operation, as follows:

\begin{equation}
\label{eq:shiftoperationgeneraldefinition}
\sigma_{(a,b,\ldots)(l,\ldots)}(x) = ROTR^a(x)\oplus ROTR^b(x)\oplus\cdots\oplus SHR^l(x)\oplus\cdots
\end{equation}

Where $(a,b,\dots)$ is the list of $ROTR$ superscripts and $(l,\ldots)$ is the list of $SHR$ superscripts.
For notational convenience, we can abbreviate the shift type with a greek letter, for example:

\begin{equation}
\mu = (a,b,c,\ldots)(l,m,\ldots)
\end{equation}

Which allows us to write, for example: $\sigma_\mu(x) \equiv \sigma_{(a,b,c,\ldots)(l,m,\ldots)}(x)$.
Given a shift type $\mu=(a,b,c)(l,m)$, we also define the corresponding \textit{complementary shift type} $\tilde{\mu}$, as follows:

\begin{equation}
\tilde{\mu} = (-a,-b,-c)(-l,-m)
\label{eq:complementaryshifttype}
\end{equation}

That is, the complement negates all of the shift amounts, both for the $ROTR$ and the $SHR$ operations.

Given a shift operation $\sigma_\mu(x)$  that is invertible, we can define its inverse $\sigma_\mu^{-1}(x)$, for which: $\sigma_\mu^{-1}(\sigma_\mu(x))=\sigma_\mu(\sigma_\mu^{-1}(x))=x$.
Note that a shift function $\sigma_\mu(x)$ is invertible, only if the shift type $\mu$ has an odd number of elements (for example, $\mu=(a,b,l)()$).
For a shift function $\sigma_\mu(x)$ consisting only of $ROTR$ bit rotations, it is known that the shift function is invertible, if and only if the number of $ROTR$ operations is odd \cite{Rivest01012011,Thomsen2008}.
For a shift function $\sigma_\nu(x)$ that includes $SHR$ bit shifts, the situation is less clear: it appears that $\sigma_\nu(x)$ having an odd number of elements is a necessary, but not a sufficient condition to ensure invertibility.

For the purpose of finding the inverse, the following linearity property of $\sigma_\mu(x)$ is useful:

\begin{equation}
\sigma_\mu(x\oplus y) = \sigma_\mu(x)\oplus\sigma_\mu(y)
\end{equation}

This linearity property also holds for the inverse of $\sigma_\mu(x)$:

\begin{equation}
\sigma_\mu^{-1}(x\oplus y) = \sigma_\mu^{-1}(x)\oplus\sigma_\mu^{-1}(y)
\end{equation}

Now, given an arbitrary $w$-bit word $x=\sum_{j=0}^{w-1}\,x_j\,2^j$, we can write the inverse shift function as:

\begin{equation}
\sigma_\mu^{-1}(x) = \sigma_\mu^{-1}\left(\sum_{j=0}^{w-1}\,x_j\,2^j\right)
  = \bigoplus_{j=0}^{w-1}\,x_j\,\sigma_\mu^{-1}(2^j)
\end{equation}

Using the notation $\sigma_{\mu,i}^{-1}(y)$ to reference the individual bits $0\le i < w$ returned by $\sigma_\mu^{-1}$, we can define the inverse matrix $T_{ij}$, as follows:

\begin{equation}
T_{ij} = \sigma_{\mu,i}^{-1}(2^j)
\label{eq:inverseshiftmatrix}
\end{equation}

Which enables us to write the inverse shift function as:

\begin{equation}
\sigma_{\mu,i}^{-1}(x) = \bigoplus_{j=0}^{w-1}\,T_{ij}\,x_j
\end{equation}

Hence, to get the inverse function $\sigma_\mu^{-1}(x)$, we just need to solve for the matrix elements $T_{ij}$, defined by equation \ref{eq:inverseshiftmatrix}.
This problem can be solved using standard linear algebra techniques, such as Euclid's extended algorithm on input polynomials \cite{Rivest01012011}, or using the Z3 Prover from Microsoft Research \cite{OrsonPeters2021,Z3Prover}, which is the approach we adopted in our code examples.

Before we can calculate the reciprocal transform of the shift function, we need the following lemma.

\begin{lemma}
\label{lem:shiftlemma}
Given registers $x$ and $\kappa$, which are both of size $w$, a general shift function $\sigma_\mu(x)$, and using the definition of the dot product from equation \ref{eq:hadamarddotproduct}, the following identity holds:
\[ \kappa\cdot\sigma_\mu(x) = \sigma_{\tilde{\mu}}(\kappa)\cdot x \]
Where $\tilde{\mu}$ is the complement of the shift type $\mu$ (see equation \ref{eq:complementaryshifttype}).
\end{lemma}

\begin{proof}
From our definition of the Hadamard dot product, we have:
\begin{align*}
\kappa\cdot\sigma_\mu(x) &= \frac{1}{2}\,\bigoplus_{j=0}^{w-1}\,\kappa_j\,\sigma_{\mu,j}(x) \\
  &= \frac{1}{2}\,\bigoplus_{j=0}^{w-1}\,\kappa_j\,
  \left( ROTR^a_j(x)\oplus\cdots\oplus SHR^\ell_j(x)\oplus\cdots \right) \\
  &= \frac{1}{2}\,\bigoplus_{j=0}^{w-1}\,
  \kappa_j\,x_{j+a\textrm{ mod }w}\oplus\cdots\oplus\kappa_j\,x_{j+\ell}\,\theta(w-1-(j+\ell))\,\theta(j+\ell)\oplus\cdots \\
  &= \frac{1}{2}\,\bigoplus_{j=0}^{w-1}\,
  \kappa_{j-a\textrm{ mod }w}\,x_j\oplus\cdots\oplus\kappa_{j-\ell}\,x_j\,\theta(w-1-(j-\ell))\,\theta(j-\ell)\oplus\cdots \\
  &= \frac{1}{2}\,\bigoplus_{j=0}^{w-1}\,
  \left( ROTR^{-a}_j(\kappa)\oplus\cdots\oplus SHR^{-\ell}_j(\kappa)\oplus\cdots \right) x_j \\
  &= \sigma_{\tilde{\mu}}(\kappa)\cdot x
\end{align*}
\end{proof}

We are now ready to calculate the reciprocal transform of the shift function $\sigma_\mu(x)$.
Using the definition of the reciprocal transform (equation \ref{eq:reciprocaltransform}), we have:

\begin{align*}
R[\sigma_\mu(x)] &= \frac{1}{2^w}\,\sum_\kappa\sum_x\sum_k\,
  e^{i2\pi\left( -\kappa\cdot\sigma_\mu(x) + x\cdot k \right)}
  \,\ket{\kappa}\bra{k} \\
&= \frac{1}{2^w}\,\sum_\kappa\sum_x\sum_k\,
  e^{i2\pi\, x\cdot(k - \sigma_{\tilde{\mu}}(\kappa))}
  \,\ket{\kappa}\bra{k} \\
&= \sum_\kappa\sum_k\,
  \delta\left(k - \sigma_{\tilde{\mu}}(\kappa)\right)
  \,\ket{\kappa}\bra{k} \\
&= \sum_\kappa\sum_k\,
  \delta\left(\sigma^{-1}_{\tilde{\mu}}(k) - \kappa\right)
  \,\ket{\kappa}\bra{k} \\
\end{align*}

So now we can calculate the values of the $\kappa$ qubits from the $k$ qubits using the relation:

\begin{equation}
\kappa = \sigma^{-1}_{\tilde{\mu}}(k)
\end{equation}

Or in terms of the inverse matrix $T_{ij}$:

\begin{equation}
\kappa_i = \bigoplus_{j=0}^{w-1}\, T_{ij} k_j \>\>\>\>
  \>\>\>\> \textrm{ where } T_{ij} = \sigma^{-1}_{\tilde{\mu},i}(2^j)
\label{eq:recipbitshiftinversematrix}
\end{equation}

To see how this works in practice, consider how to construct the circuit for the reciprocal transform of the following bit shifter function that operates on a 4-bit word (that is, $w=4$):

\begin{equation}
\sigma_{(0,1)(3)}(x) = ROTR^0(x)\oplus ROTR^1(x)\oplus SHR^3(x)
\end{equation}

The complementary shift type is then $\tilde{\mu}=(0,-1)(-3)$ and, using a linear algebra package, we find that the inverse function $\sigma^{-1}_{(0,-1)(-3)}(x)$ is given by the following table of binary values (implicitly defining the inverse matrix $T_{ij}$):

\begin{align*}
\sigma^{-1}_{(0,-1)(-3)}(2^0) &= 0111 \\
\sigma^{-1}_{(0,-1)(-3)}(2^1) &= 1001 \\
\sigma^{-1}_{(0,-1)(-3)}(2^2) &= 1011 \\
\sigma^{-1}_{(0,-1)(-3)}(2^3) &= 1111 \\
\end{align*}

The reciprocal bit shifter circuit (see figure \ref{fig:recipbitshiftergate}) can now be constructed as follows:

\begin{itemize}
\item The circuit requires 8 qubits: 4 qubits to hold the incoming values, $k_0, k_1, k_2, k_3$, and another 4 ancillary qubits (initially set to $\ket{0}$) to temporarily hold the transformed values, $\kappa_0, \kappa_1, \kappa_2, \kappa_3$.

\item For every non-zero entry in the inverse matrix $T_{ij}$, we see from equation \ref{eq:recipbitshiftinversematrix} that we need to insert a CNOT gate, which is applied to ancillary qubit $\kappa_i$ and controlled by qubit $k_j$.

\item By calculating $\sigma_{(0,-1)(-3)}(\kappa)$ and XORing the result to the $k_j$ qubits, we uncompute the $k_j$ qubits to zero.

\item Swap the corresponding $k_j$ and ancillary qubits, so that the (originally) $k_j$ qubits now contain the $\kappa_j$ result and the ancillaries are reset to $\ket{0}$.
\end{itemize}

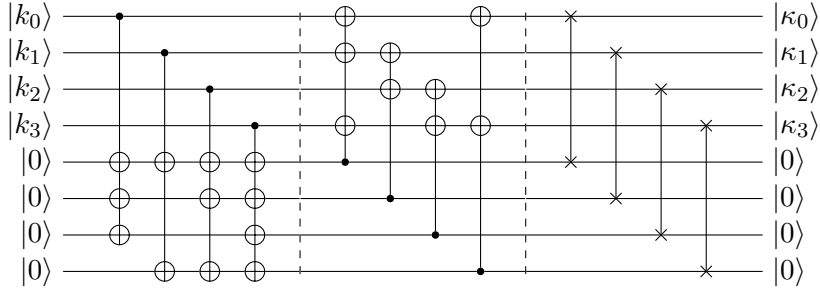
\begin{figure}[H]
\centering
\begin{tikzpicture}
    \begin{yquant}
        qubit {$\ket{\reg_{\idx}}$} k[4];
        qubit {$\ket{0}$} anc[4];
        hspace {4mm} -;

        cnot anc[0, 1, 2] | k[0];
        cnot anc[0, 3]    | k[1];
        cnot anc[0, 1, 3] | k[2];
        cnot anc[0, 1, 2, 3] | k[3];
        barrier (k, anc);
        cnot k[0, 1, 3] | anc[0];
        cnot k[1, 2] | anc[1];
        cnot k[2, 3] | anc[2];
        cnot k[3, 0] | anc[3];
        barrier (k, anc);
        swap (k[0], anc[0]);
        swap (k[1], anc[1]);
        swap (k[2], anc[2]);
        swap (k[3], anc[3]);

        hspace {4mm} -;
        output {$\ket{\kappa_{\idx}}$} k[0, 1, 2, 3];
        output {$\ket{0}$} anc[0, 1, 2, 3];
    \end{yquant}
\end{tikzpicture}
\caption{Reciprocal bit shifter gate for $\sigma_{(0,1)(3)}(x)$}\label{fig:recipbitshiftergate}
\end{figure}

\section{Simulated hash algorithm examples}

\subsection{Simulating a simple operation chain}
\label{simulatesimplechain}

From the chain rule theorem (section \ref{chainruletheorem}), we know that we can construct an oracle function by chaining multiple operations together.
For our first example, we consider a simple chain of two operations: modulo addition (adding two 4-bit integers), followed by a shift operation.
Given two 4-bit qubit registers $x$ and $y$, we therefore consider the function $g(x, y)$ mapping $(x, y)\mapsto(x^\prime,y^\prime)$ defined, as follows:

\begin{equation}
g(x, y) = \left(x, \,\sigma_{(0, 1, 3)()}(x + y)\right) = (x^\prime,y^\prime)
\end{equation}

Where $\sigma_{(0, 1, 3)()}$ is a shift function defined using the syntax from equation \ref{eq:shiftoperationgeneraldefinition} and the $+$ symbol is defined as addition modulo $2^4$.
Given the target values $x^\prime_t$ and $y^\prime_t$, we can then define the oracle function $f(x, y)$ as:

\begin{equation}
f(x, y) = \left(x\oplus x^\prime_t, \,\sigma_{(0, 1, 3)()}(x + y)\oplus y^\prime_t \right)
\end{equation}

Which ensures that $g(x,y)=(x^\prime_t,y^\prime_t)$ when $f(x,y)=0$.

This problem is implemented by the code shown in appendix \ref{appx:simpleoperationchainexample} (listing \ref{lst:simplechainexample}) \cite{QFrameExamples}, which uses the partial oracles algorithm (section \ref{partialoraclealgorithm}) to find the value(s) of $x$ and $y$ that satisfy the given constraints.
Consider the following run of this algorithm, which runs on the default Qrisp simulator backend:

\begin{itemize}
\item Registers $x$ and $y$ are initialized in a uniform superposition state: $\ket{\Psi} = \sum_x\sum_y\,\ket{x}\ket{y}$.
\item Target values are: $(x^\prime_t, y^\prime_t) = (4, 1)$.
\end{itemize}

After a single partial oracle iteration, the simulation returns the solution $(x_s, y_s)=(4, 7)$ satisfying $g(4, 7)=(4,1)$, with 100\% probability.
By contrast, an equivalent implementation of this problem using Grover's algorithm would require $\sqrt{2^8}=16$ iterations.

We can run this simulation with every possible combination of target values $(x^\prime_t, y^\prime_t)$ and verify that it gives the correct solution in each case.

\subsection{QFrame library}
\label{qframelibrary}

When we consider the code in listing \ref{lst:simplechainexample}, we see that for any oracle function, we need to implement three chunks of code in parallel:

\begin{itemize}
\item \textit{Calculation of the classical result of the target function}---given a particular function $g(x,y)$ that we are investigating, we need to be able to calculate $g(x, y)=(x_t, y_t)$ so that we can check the result of the search problem.

\item \textit{Implemention of the circuit for the oracle function}---based on the function $g(x,y)$ and given the target values $(x_t, y_t)$, we must implement the circuit for the oracle function $f(x, y)$.

\item \textit{Implemention of the circuit for the reciprocal oracle function}---we also need to implement the circuit for the reciprocal transform of the oracle function $R[f(x, y)]$.
\end{itemize}

These three chunks of code are all derived from the same function of interest $g(x,y)$, but writing the code chunks manually and separately is both tedious and error-prone, especially for functions that are long and complex.
It would be much more practical and reliable to specify the function algorithm once and then generate the code chunks directly from the specification.

For this purpose, we have implemented a small utility library, the QFrame library \cite{QFrame_GitHub,QFrame_PyPI} layered over Eclipse Qrisp, which supports Python syntax for defining the function of interest $g(\cdot)$.
The \lstinline{qframe.QFrameUInt} type provides an abstract representation of the qubit registers and QFrame provides a limited set of functions and operations for combining the \lstinline{QFrameUInt} instances.

The following operations and functions are provided by QFrame:

\begin{itemize}
\item $\pluseq$ operation---defines in-place addition modulo $2^w$, where $w$ is the bit-width of the \lstinline{QFrameUInt} arguments.
The following cases are supported:

\begin{itemize}
\item Given \lstinline{QFrameUInt} instances $x$ and $y$, the operation $x \pluseq y$ implements $P\,\ket{x}\ket{y} = \ket{x+y}\ket{y}$.
\item Given a \lstinline{QFrameUInt} instance $x$ and an integer constant $c$, the operation $x \pluseq c$ implements $\ket{x}\mapsto\ket{x+c}$.
\item Given a \lstinline{QFrameUInt} instance $x$ and a temporary function $t(y)$, such that $V\,\ket{y} = \ket{t(y)}$, the operation $x \mathrel{+}= t(y)$ implements $V^\dagger PV \ket{x}\ket{y} = \ket{x + t(y)}\ket{y}$.
\end{itemize}

\item \lstinline{qframe.maj(a, b, c)}---defines the majority function $Maj(a, b, c)$, where $a$, $b$, and $c$ are $w$-bit registers (see section \ref{majorityfunction}).
The \lstinline{maj(a, b, c)} function is a temporary function, in the sense that it does not permanently change the values of registers $a$, $b$, and $c$.
It is intended to be used on the right-hand side of a $\pluseq$ expression.

\item \lstinline{qframe.ch(a, b, c)}---defines the choose function $Ch(a, b, c)$, where $a$, $b$, and $c$ are $w$-bit registers (see section \ref{choosefunction}).
The \lstinline{ch(a, b, c)} function is a temporary function, in the sense that it does not permanently change the values of registers $a$, $b$, and $c$.
It is intended to be used on the right-hand side of a $\pluseq$ expression.

\item \lstinline{qframe.Rotr} class---can be used to define shift operator instances.
For example, \lstinline{r = qframe.Rotr(4, [0, 1], shr_list=[3])} defines a 4-bit wide shift operator of type $\mu=(0,1)(3)$ (see \ref{eq:shiftoperationgeneraldefinition}).
The \lstinline{r.shift_inline(x)} method can then be invoked to apply the shift function to register $x$.
If you need to add the result of a shift operation to another register in a $\pluseq$ expression, you should invoke \lstinline{r.shift(x)}, which acts as a temporary function (leaving register $x$ unchanged after the addition operation).

\end{itemize}

This utility provides all of the operations and functions you need to implement the SHA-256 algorithm.
After defining an algorithm with QFrame, the \lstinline{QFrameSession} object provides access to the following generated methods:

\begin{itemize}
\item \lstinline{QFrameSession.calculate(input_args)}---implements the classical function $g(\cdot)$, taking the given \lstinline{input_arg} values for the registers and returning the values of the registers after calculating $g(\cdot)$.

\item \lstinline{QFrameSession.apply_oracle_gate(target_dict)}---generates the circuit of the oracle function, where the relevant target values are specified by the \lstinline{target_dict} argument.

\item \lstinline{QFrameSession.apply_recip_oracle_gate()}---generates the circuit for the reciprocal oracle function, $R[f(\cdot)]$.

\end{itemize}

Instead of invoking \lstinline{apply_oracle_gate()} and \lstinline{apply_recip_oracle_gate()} directly, we also have the option of calling the \lstinline{QFrameSession.partial_oracle_iteration(target_dict)} method, which generates the complete code for the partial oracles algorithm (except for the initialization of the qubit registers).

\subsection{Toy hash algorithm}
\label{toyhashalgorithm}

With the computing resources currently available, it is neither possible to solve a full hash algorithm on a quantum computer nor is it possible to simulate it on a conventional computer.
To investigate the feasibility of solving hash algorithms on a quantum computer, therefore, it is necessary to work with a scaled-down algorithm, using fewer registers, smaller register bit-width, and fewer operations.
In this section, we present such a scaled-down (or toy) hash algorithm and we apply the partial oracles algorithm to investigate how effective it can be at inverting the hash algorithm.

Our goal here is not to define an algorithm that has any meaningful usefulness as a hash algorithm, but simply to define an algorithm that uses similar operations and has a similar sequence of steps to the SHA-256 algorithm.
The toy hash algorithm should also be small enough that it can run on machines with modest hardware requirements.
We define the toy hash algorithm, as follows:

\begin{itemize}
\item The message schedule is represented by a single 4-qubit register, $W_0$.
\item There are four working variables $a$, $b$, $c$, and $d$, each 4-qubits wide.
\item The addition operator $+$ is defined as addition modulo 16.
\item There are two shift operators (see equation \ref{eq:shiftoperationgeneraldefinition}), defined as follows:

\begin{align*}
\Sigma(x) &= \sigma_{(0, 1, 3)()}(x) \\
\sigma(x) &= \sigma_{(0, 1)(3)}(x) \\
\end{align*}

\item The constants $K_0$ to $K_{15}$ are got by taking the first sixteen SHA-256 constants from the SHS standard $K^{\{256\}}_0$ to $K^{\{256\}}_{15}$ and truncating them to the lowest 4 bits. 

\item For $t=0$ to $3$:

\begin{align*}
T_1 &= d + \Sigma(a) + Ch(a, b, c) + K_t + W_0 \\
T_2 &= Maj(a,\, b + T_1,\, c) \\
d &= c \\
c &= b + T_1 \\
b &= a \\
a &= T_1 + T_2 \\
W_0 &= \sigma(W_0) \\
\end{align*}

\end{itemize}

After applying these steps, the state of the registers changes from $(a,b,c,d,W_0)$ to $(a^\prime,b^\prime,c^\prime,d^\prime,W^\prime_0)$, which we write as $g(a,b,c,d,W_0)=(a^\prime,b^\prime,c^\prime,d^\prime,W^\prime_0)$.

This toy hash algorithm is implemented by the code from appendix \ref{appx:toyshaexample} (listing \ref{lst:toyhashexample}) \cite{QFrameExamples}, using the QFrame library to define the hash algorithm and the Qrisp client library to perform the simulation.
Consider the following run of this algorithm, which runs on the default Qrisp simulator backend:

\begin{itemize}
\item Registers $a$, $b$, $c$, $d$ and $W_0$ are initialized in a uniform superposition state: $\ket{\Psi} = \sum_a\sum_b\sum_c\sum_d\sum_{W_0}\,\ket{a}\ket{b}\ket{c}\ket{d}\ket{W_0}$.
\item Target values are: $(a^\prime_t, b^\prime_t, c^\prime_t, d^\prime_t, W^\prime_{0,t}) = (13, 1, 7, 4, 10)$.
\end{itemize}

After a single partial oracle iteration, the simulation returns the solution $(a_s, b_s, c_s, d_s, W_{0,s}) = (7, 5, 2, 10, 8)$, satisfying $g(7, 5, 2, 10, 8)=(13, 1, 7, 4, 10)$, with 100\% probability.
By contrast, an equivalent implementation of this problem using Grover's algorithm would require $\sqrt{2^{20}}=1024$ iterations.

Note that this problem is closely related to the problem of breaking a hash algorithm, but is different in important respects.
For the problem of breaking a hash algorithm, you would usually be given the input hash value (for example, the initial values of registers $a$, $b$, $c$, and $d$) and the calculated output hash value (for example, the target values of the registers $a^\prime,b^\prime,c^\prime,d^\prime$), and you would then use this information to recover part or all of the input message ($W_0$).
For this type of problem, you need to impose constraints on the output values \textit{and} on the input values of the hash calculation.
But the quantum search algorithm we have described so far only enables you to apply constraints on the output values of the problem.

\section{Conclusion}

In this work, we introduced an operator---the reciprocal transform of the oracle function $R[f(x)]$ (see equation \ref{eq:reciprocaltransform})---which effectively replaces the second Grover operator in a new search algorithm.
By combining the concept of partial oracles with the reciprocal transform operator, we proved the validity of the partial oracles algorithm (section \ref{partialoraclealgorithm}) for the case of a bijective oracle function.
This algorithm reduces the size of the search space by $1/2$ for each iteration, thereby returning the search result within $n = \log(N)$ iterations.
Furthermore, we observed that the iterations of the partial oracles algorithm can be executed in parallel (see section \ref{parallelization}) instead of sequentially, so that the $n$ search iterations can be compressed into a single partial oracle iteration.

We proved the chain rule for the reciprocal transform (equation \ref{eq:chainruletheorem} in section \ref{chainruletheorem}), which enables us to break down the transformation of complex oracle functions $f(x)$ into a sequence of manageable steps.
We also demonstrated that the target value of the search can effectively be ignored when calculating the reciprocal transform operator.

To demonstrate the effectiveness of the partial oracles algorithm in practice, we calculated the reciprocal transforms and associated circuits for the following elementary operations:

\begin{itemize}
\item $+$ ---addition modulo $2^w$, where $w$ is the bit-width of the summands.
\item $Ch(a, b, c)$ ---choose function (see section \ref{choosefunction}).
\item $Maj(a, b, c)$ ---majority function (see section \ref{majorityfunction}).
\item $\sigma_{(a,b,c,\ldots)(l,m,\ldots)}(x)$ ---bit shifting functions (see section \ref{bitshiftingfunctions}).
\end{itemize}

We combine these elementary operations to define hash functions, which can then be inverted using the partial oracles algorithm (but not in a way that could be used to reveal the input message, given the hash value).
In particular, we define a toy SHA algorithm using 20 qubits and we demonstrate that we can find the input value that corresponds to a specific target output value.
Effectively, out of 1,048,576 values in the initial quantum state, the partial oracles algorithm finds the unique solution by executing one partial oracle iteration.

We note that the version of the algorithm presented here is limited to searching with oracle functions constructed using only in-place operations.
The next stage of development of the partical oracles algorithm will be to remove this limitation by expanding our analysis to out-of-place operations.
This problem will be addressed in a forthcoming paper (part II).

\appendix

\section{Simple operation chain example}
\label{appx:simpleoperationchainexample}

The following code example implements a simple use case involving partial oracles search where the oracle function performs an integer addition followed by a bit shift operation.
The original source code for this example is available online \cite{QFrameExamples}.

\begin{lstlisting}[language=Python, caption=Simple operation chain example, label=lst:simplechainexample]
import qrisp
from qrisp import h

from qframe import recip_adder_gate, Rotr


def phase_oracle_init(p):
    qrisp.z(qrisp.h(p))

def apply_oracle_function(r: Rotr, x:qrisp.QuantumFloat, y:qrisp.QuantumFloat, anc_y:qrisp.QuantumFloat, target: tuple):
    y += x
    r.rotr_gate(y, anc_y)

    # Set target value(s)
    x_target, y_target = target
    for i in range(x_target.bit_length()):
        if (x_target & 1<<i): qrisp.x(x[i])
    for i in range(y_target.bit_length()):
        if (y_target & 1<<i): qrisp.x(y[i])

def apply_recip_oracle_function(r: Rotr, x:qrisp.QuantumFloat, y:qrisp.QuantumFloat, anc_y:qrisp.QuantumFloat, c:qrisp.Qubit, anc:qrisp.Qubit):
    # Initialize reciprocal carry bit and phase oracle bit
    h(c)
    phase_oracle_init(anc)
    # Apply reciprocal circuit
    recip_adder_gate(x, y, c, anc)
    r.recip_rotr_gate(y, anc_y)


def run_simulation(x_seed, y_seed):
    # Example parameters
    width = 4
    r = Rotr(width, [0, 1, 3])

    # Calculate target tuple from seed values: x_seed and y_seed
    target  = (x_seed, r.rotr_function((x_seed + y_seed) % 2**width))
    print(f'\n----------------------------------------------------------')
    print(f'calculate(x_seed = {x_seed}, y_seed = {y_seed}) = {target}')

    # Initialize variables
    x = qrisp.QuantumFloat(width, name='x')
    y = qrisp.QuantumFloat(width, name='y')
    carry = qrisp.QuantumVariable(1, name='c')
    anc   = qrisp.QuantumVariable(1, name='anc')
    anc_y = qrisp.QuantumFloat(width, name='anc_y')

    # Prepare the quantum state in an equal-weighted superposition (Walsh-Hadamard transform)
    h(x)
    h(y)

    qrisp.barrier(x[:] + y[:] + carry[0] + anc[0])

    # Single partial oracle iteration
    with qrisp.conjugate(apply_oracle_function)(r, x, y, anc_y, target):
        qrisp.barrier(x[:] + y[:] + carry[0] + anc[0])
        qrisp.s(x)
        qrisp.s(y)
        qrisp.barrier(x[:] + y[:] + carry[0] + anc[0])
    qrisp.barrier(x[:] + y[:] + carry[0] + anc[0])
    h(x)
    h(y)
    qrisp.barrier(x[:] + y[:] + carry[0] + anc[0])
    with qrisp.conjugate(apply_recip_oracle_function)(r, x, y, anc_y, carry[0], anc[0]):
        qrisp.barrier(x[:] + y[:] + carry[0] + anc[0])
        qrisp.s(x)
        qrisp.s(y)
        qrisp.barrier(x[:] + y[:] + carry[0] + anc[0])
    h(y)
    h(x)

    # Show the circuit
    print(y.qs)

    # Show result
    result_dict = qrisp.multi_measurement([x, y])
    print(f'\nResult: {result_dict}')

    for result_tuple in result_dict:
        assert (x_seed, y_seed) == result_tuple

# Run the simulation with seed values
run_simulation(4, 7)

# Test all possible seed values
# for x_seed in range(16):
#     for y_seed in range(16):
#         run_simulation(x_seed, y_seed)
\end{lstlisting}

\section{Toy SHA example}
\label{appx:toyshaexample}

The following code example implements a use case involving partial oracles search where the oracle function performs a simplified "toy" hash operation on a 4-bit input message.
The original source code for this example is available online \cite{QFrameExamples}.

\begin{lstlisting}[language=Python, caption=Toy hash algorithm example, label=lst:toyhashexample]
import qrisp
import qrisp.jasp
from qrisp import h

import qframe
from qframe.core.qframe_uint import QFrameUInt
from qframe import Rotr
import sys


# SHA-256 constants
K = [
    0x428a2f98, 0x71374491, 0xb5c0fbcf, 0xe9b5dba5, 0x3956c25b, 0x59f111f1, 0x923f82a4, 0xab1c5ed5,
    0xd807aa98, 0x12835b01, 0x243185be, 0x550c7dc3, 0x72be5d74, 0x80deb1fe, 0x9bdc06a7, 0xc19bf174,
    0xe49b69c1, 0xefbe4786, 0x0fc19dc6, 0x240ca1cc, 0x2de92c6f, 0x4a7484aa, 0x5cb0a9dc, 0x76f988da,
    0x983e5152, 0xa831c66d, 0xb00327c8, 0xbf597fc7, 0xc6e00bf3, 0xd5a79147, 0x06ca6351, 0x14292967,
    0x27b70a85, 0x2e1b2138, 0x4d2c6dfc, 0x53380d13, 0x650a7354, 0x766a0abb, 0x81c2c92e, 0x92722c85,
    0xa2bfe8a1, 0xa81a664b, 0xc24b8b70, 0xc76c51a3, 0xd192e819, 0xd6990624, 0xf40e3585, 0x106aa070,
    0x19a4c116, 0x1e376c08, 0x2748774c, 0x34b0bcb5, 0x391c0cb3, 0x4ed8aa4a, 0x5b9cca4f, 0x682e6ff3,
    0x748f82ee, 0x78a5636f, 0x84c87814, 0x8cc70208, 0x90befffa, 0xa4506ceb, 0xbef9a3f7, 0xc67178f2
]
# Initial hash value
H0 = [0x6a09e667, 0xbb67ae85, 0x3c6ef372, 0xa54ff53a, 0x510e527f, 0x9b05688c, 0x1f83d9ab, 0x5be0cd19]


width = 4
width_mask = (0b1 << width) - 1
a = QFrameUInt(width, name='a')
b = QFrameUInt(width, name='b')
c = QFrameUInt(width, name='c')
d = QFrameUInt(width, name='d')
W0 = QFrameUInt(width, name='W0')

big_sigma = Rotr(width, [0, 1, 3])
sigma     = Rotr(width, [0, 1], shr_list=[3])

# Define algorithm using QFrame
abcd = [a, b, c, d]  # Working variables

a_idx = 0
b_idx = 1
c_idx = 2
d_idx = 3
for t in range(4):
    # Update the working variables
    abcd[d_idx] += big_sigma.shift(abcd[a_idx])
    abcd[d_idx] += qframe.ch(abcd[a_idx], abcd[b_idx], abcd[c_idx])
    abcd[d_idx] += K[t] & width_mask
    abcd[d_idx] += W0
    abcd[b_idx] += abcd[d_idx]   # Also make an update to b
    abcd[d_idx] += qframe.maj(abcd[a_idx], abcd[b_idx], abcd[c_idx])
    # Update the message schedule variable W0
    sigma.shift_inline(W0)
    # Finally, update indices for the next round
    a_idx = (a_idx - 1) % 4
    b_idx = (b_idx - 1) % 4
    c_idx = (c_idx - 1) % 4
    d_idx = (d_idx - 1) % 4


# Get the QFrameSession object
qfs = a.qfs

# Prepare the target state
a_seed = H0[0] & width_mask
b_seed = H0[1] & width_mask
c_seed = H0[2] & width_mask
d_seed = H0[3] & width_mask
W0_seed = ord('H') & width_mask   # = 8 (for a 4-bit register)

seed_args = {a: a_seed, b: b_seed, c: c_seed, d: d_seed, W0: W0_seed}
target = qfs.calculate(seed_args, raw_result=True)
print(f'\ncalculate(a: {a_seed}, b: {b_seed}, c: {c_seed}, d: {d_seed}, W0: {W0_seed}) = {qfs.calculate(seed_args)}\n')

# Prepare the quantum state in an equal-weighted superposition (Walsh-Hadamard transform)
h(a.qv)
h(b.qv)
h(c.qv)
h(d.qv)
h(W0.qv)

qfs.partial_oracle_iteration(target)

# Show the circuit
print(a.qv.qs)

# Show result
result_dict = qrisp.multi_measurement([a.qv, b.qv, c.qv, d.qv, W0.qv])

print(f'\nResult: {result_dict}')
print(f'Calculated: calculate(a: {a_seed}, b: {b_seed}, c: {c_seed}, d: {d_seed}, W0: {W0_seed}) = {qfs.calculate(seed_args)}\n')
\end{lstlisting}

\bibliography{qframe}

\begin{thebibliography}{17}%
\makeatletter
\providecommand \@ifxundefined [1]{%
 \@ifx{#1\undefined}
}%
\providecommand \@ifnum [1]{%
 \ifnum #1\expandafter \@firstoftwo
 \else \expandafter \@secondoftwo
 \fi
}%
\providecommand \@ifx [1]{%
 \ifx #1\expandafter \@firstoftwo
 \else \expandafter \@secondoftwo
 \fi
}%
\providecommand \natexlab [1]{#1}%
\providecommand \enquote  [1]{``#1''}%
\providecommand \bibnamefont  [1]{#1}%
\providecommand \bibfnamefont [1]{#1}%
\providecommand \citenamefont [1]{#1}%
\providecommand \href@noop [0]{\@secondoftwo}%
\providecommand \href [0]{\begingroup \@sanitize@url \@href}%
\providecommand \@href[1]{\@@startlink{#1}\@@href}%
\providecommand \@@href[1]{\endgroup#1\@@endlink}%
\providecommand \@sanitize@url [0]{\catcode `\\12\catcode `\$12\catcode
  `\&12\catcode `\#12\catcode `\^12\catcode `\_12\catcode `\%12\relax}%
\providecommand \@@startlink[1]{}%
\providecommand \@@endlink[0]{}%
\providecommand \url  [0]{\begingroup\@sanitize@url \@url }%
\providecommand \@url [1]{\endgroup\@href {#1}{\urlprefix }}%
\providecommand \urlprefix  [0]{URL }%
\providecommand \Eprint [0]{\href }%
\providecommand \doibase [0]{https://doi.org/}%
\providecommand \selectlanguage [0]{\@gobble}%
\providecommand \bibinfo  [0]{\@secondoftwo}%
\providecommand \bibfield  [0]{\@secondoftwo}%
\providecommand \translation [1]{[#1]}%
\providecommand \BibitemOpen [0]{}%
\providecommand \bibitemStop [0]{}%
\providecommand \bibitemNoStop [0]{.\EOS\space}%
\providecommand \EOS [0]{\spacefactor3000\relax}%
\providecommand \BibitemShut  [1]{\csname bibitem#1\endcsname}%
\let\auto@bib@innerbib\@empty
\bibitem [{\citenamefont {Hoefler}\ \emph {et~al.}(2023)\citenamefont
  {Hoefler}, \citenamefont {Haener},\ and\ \citenamefont
  {Troyer}}]{hoefler2023disentanglinghypepracticalityrealistically}%
  \BibitemOpen
  \bibfield  {author} {\bibinfo {author} {\bibfnamefont {T.}~\bibnamefont
  {Hoefler}}, \bibinfo {author} {\bibfnamefont {T.}~\bibnamefont {Haener}},\
  and\ \bibinfo {author} {\bibfnamefont {M.}~\bibnamefont {Troyer}},\ }\href
  {https://arxiv.org/abs/2307.00523} {\bibinfo {title} {Disentangling hype from
  practicality: On realistically achieving quantum advantage}} (\bibinfo {year}
  {2023}),\ \Eprint {https://arxiv.org/abs/2307.00523} {arXiv:2307.00523
  [quant-ph]} \BibitemShut {NoStop}%
\bibitem [{\citenamefont {Stoudenmire}\ and\ \citenamefont
  {Waintal}(2023)}]{Stoudenmire2023GroversAO}%
  \BibitemOpen
  \bibfield  {author} {\bibinfo {author} {\bibfnamefont {E.~M.}\ \bibnamefont
  {Stoudenmire}}\ and\ \bibinfo {author} {\bibfnamefont {X.}~\bibnamefont
  {Waintal}},\ }\bibfield  {title} {\bibinfo {title} {Grover's algorithm offers
  no quantum advantage},\ }\href {https://doi.org/10.48550/arXiv.2303.11317}
  {\bibfield  {journal} {\bibinfo  {journal} {arXiv preprint arXiv.2303.11317}\
  } (\bibinfo {year} {2023})}\BibitemShut {NoStop}%
\bibitem [{\citenamefont {Stoudenmire}\ and\ \citenamefont
  {Waintal}(2024)}]{PhysRevX.14.041029}%
  \BibitemOpen
  \bibfield  {author} {\bibinfo {author} {\bibfnamefont {E.~M.}\ \bibnamefont
  {Stoudenmire}}\ and\ \bibinfo {author} {\bibfnamefont {X.}~\bibnamefont
  {Waintal}},\ }\bibfield  {title} {\bibinfo {title} {Opening the black box
  inside grover's algorithm},\ }\href
  {https://doi.org/10.1103/PhysRevX.14.041029} {\bibfield  {journal} {\bibinfo
  {journal} {Phys. Rev. X}\ }\textbf {\bibinfo {volume} {14}},\ \bibinfo
  {pages} {041029} (\bibinfo {year} {2024})}\BibitemShut {NoStop}%
\bibitem [{\citenamefont {Bolton}(2024)}]{PartialOraclesArticle2024}%
  \BibitemOpen
  \bibfield  {author} {\bibinfo {author} {\bibfnamefont {F.~M.}\ \bibnamefont
  {Bolton}},\ }\href {https://arxiv.org/abs/2403.13035} {\bibinfo {title}
  {Accelerated quantum search using partial oracles and grover's algorithm}}
  (\bibinfo {year} {2024}),\ \Eprint {https://arxiv.org/abs/2403.13035}
  {arXiv:2403.13035 [quant-ph]} \BibitemShut {NoStop}%
\bibitem [{\citenamefont {Galindo}\ and\ \citenamefont
  {Martin-Delgado}(2000)}]{Galindo2000FamilyOG}%
  \BibitemOpen
  \bibfield  {author} {\bibinfo {author} {\bibfnamefont {A.}~\bibnamefont
  {Galindo}}\ and\ \bibinfo {author} {\bibfnamefont {M.~A.}\ \bibnamefont
  {Martin-Delgado}},\ }\bibfield  {title} {\bibinfo {title} {Family of
  grover’s quantum-searching algorithms},\ }\href
  {https://api.semanticscholar.org/CorpusID:8335488} {\bibfield  {journal}
  {\bibinfo  {journal} {Physical Review A}\ }\textbf {\bibinfo {volume} {62}},\
  \bibinfo {pages} {062303} (\bibinfo {year} {2000})}\BibitemShut {NoStop}%
\bibitem [{\citenamefont {Grover}(1996)}]{LKGrover1996}%
  \BibitemOpen
  \bibfield  {author} {\bibinfo {author} {\bibfnamefont {L.~K.}\ \bibnamefont
  {Grover}},\ }\bibfield  {title} {\bibinfo {title} {A fast quantum mechanical
  algorithm for database search},\ }in\ \href
  {https://doi.org/10.1145/237814.237866} {\emph {\bibinfo {booktitle}
  {Proceedings of the Twenty-Eighth Annual ACM Symposium on Theory of
  Computing}}},\ \bibinfo {series and number} {STOC '96}\ (\bibinfo
  {publisher} {Association for Computing Machinery},\ \bibinfo {address} {New
  York, NY, USA},\ \bibinfo {year} {1996})\ p.\ \bibinfo {pages}
  {212–219}\BibitemShut {NoStop}%
\bibitem [{\citenamefont {Grover}(1997)}]{PhysRevLett.79.325}%
  \BibitemOpen
  \bibfield  {author} {\bibinfo {author} {\bibfnamefont {L.~K.}\ \bibnamefont
  {Grover}},\ }\bibfield  {title} {\bibinfo {title} {Quantum mechanics helps in
  searching for a needle in a haystack},\ }\href
  {https://doi.org/10.1103/PhysRevLett.79.325} {\bibfield  {journal} {\bibinfo
  {journal} {Phys. Rev. Lett.}\ }\textbf {\bibinfo {volume} {79}},\ \bibinfo
  {pages} {325} (\bibinfo {year} {1997})}\BibitemShut {NoStop}%
\bibitem [{\citenamefont {Long}(2001)}]{Long2001GroverAW}%
  \BibitemOpen
  \bibfield  {author} {\bibinfo {author} {\bibfnamefont {G.~L.}\ \bibnamefont
  {Long}},\ }\bibfield  {title} {\bibinfo {title} {Grover algorithm with zero
  theoretical failure rate},\ }\href
  {https://api.semanticscholar.org/CorpusID:5742433} {\bibfield  {journal}
  {\bibinfo  {journal} {Physical Review A}\ }\textbf {\bibinfo {volume} {64}},\
  \bibinfo {pages} {022307} (\bibinfo {year} {2001})}\BibitemShut {NoStop}%
\bibitem [{\citenamefont {{National Institute of Standards and
  Technology}}(2015)}]{SHS_FIPS_180_4}%
  \BibitemOpen
  \bibfield  {author} {\bibinfo {author} {\bibnamefont {{National Institute of
  Standards and Technology}}},\ }\href
  {https://csrc.nist.gov/pubs/fips/180-4/upd1/final} {\bibinfo {title} {{Secure
  Hash Standard}}} (\bibinfo {year} {2015})\BibitemShut {NoStop}%
\bibitem [{\citenamefont {Cuccaro}\ \emph {et~al.}(2004)\citenamefont
  {Cuccaro}, \citenamefont {Draper}, \citenamefont {Kutin},\ and\ \citenamefont
  {Moulton}}]{cuccaro2004newquantumripplecarryaddition}%
  \BibitemOpen
  \bibfield  {author} {\bibinfo {author} {\bibfnamefont {S.~A.}\ \bibnamefont
  {Cuccaro}}, \bibinfo {author} {\bibfnamefont {T.~G.}\ \bibnamefont {Draper}},
  \bibinfo {author} {\bibfnamefont {S.~A.}\ \bibnamefont {Kutin}},\ and\
  \bibinfo {author} {\bibfnamefont {D.~P.}\ \bibnamefont {Moulton}},\ }\href
  {https://arxiv.org/abs/quant-ph/0410184} {\bibinfo {title} {A new quantum
  ripple-carry addition circuit}} (\bibinfo {year} {2004}),\ \Eprint
  {https://arxiv.org/abs/quant-ph/0410184} {arXiv:quant-ph/0410184 [quant-ph]}
  \BibitemShut {NoStop}%
\bibitem [{\citenamefont {Rivest}(2011)}]{Rivest01012011}%
  \BibitemOpen
  \bibfield  {author} {\bibinfo {author} {\bibfnamefont {R.~L.}\ \bibnamefont
  {Rivest}},\ }\bibfield  {title} {\bibinfo {title} {{The invertibility of the
  XOR of rotations of a binary word}},\ }\href
  {https://doi.org/10.1080/00207161003596708} {\bibfield  {journal} {\bibinfo
  {journal} {International Journal of Computer Mathematics}\ }\textbf {\bibinfo
  {volume} {88}},\ \bibinfo {pages} {281} (\bibinfo {year} {2011})},\ \Eprint
  {https://arxiv.org/abs/https://doi.org/10.1080/00207161003596708}
  {https://doi.org/10.1080/00207161003596708} \BibitemShut {NoStop}%
\bibitem [{\citenamefont {Thomsen}(2008)}]{Thomsen2008}%
  \BibitemOpen
  \bibfield  {author} {\bibinfo {author} {\bibfnamefont {S.~S.}\ \bibnamefont
  {Thomsen}},\ }\emph {\bibinfo {title} {Cryptographic Hash Functions}},\
  \href@noop {} {\bibinfo {type} {Phd thesis}},\ \bibinfo  {school} {Technical
  University of Denmark} (\bibinfo {year} {2008})\BibitemShut {NoStop}%
\bibitem [{\citenamefont {Peters}(2021)}]{OrsonPeters2021}%
  \BibitemOpen
  \bibfield  {author} {\bibinfo {author} {\bibfnamefont {O.}~\bibnamefont
  {Peters}},\ }\href {https://stackoverflow.com/questions/66607696} {\bibinfo
  {title} {{Reverse SHA-256 sigma0 function within complexity of O(n)?}}}
  (\bibinfo {year} {2021})\BibitemShut {NoStop}%
\bibitem [{\citenamefont {{Microsoft Research}}(2026)}]{Z3Prover}%
  \BibitemOpen
  \bibfield  {author} {\bibinfo {author} {\bibnamefont {{Microsoft
  Research}}},\ }\href@noop {} {\bibinfo {title} {Z3 prover}},\ \bibinfo
  {howpublished} {\url{https://github.com/Z3Prover/z3}} (\bibinfo {year}
  {2026})\BibitemShut {NoStop}%
\bibitem [{\citenamefont {Bolton}(2026{\natexlab{a}})}]{QFrameExamples}%
  \BibitemOpen
  \bibfield  {author} {\bibinfo {author} {\bibfnamefont {F.}~\bibnamefont
  {Bolton}},\ }\href@noop {} {\bibinfo {title} {{QFrame examples}}},\ \bibinfo
  {howpublished} {\url{https://github.com/bradan-quantum/qframe-examples/}}
  (\bibinfo {year} {2026}{\natexlab{a}})\BibitemShut {NoStop}%
\bibitem [{\citenamefont {Bolton}(2026{\natexlab{b}})}]{QFrame_GitHub}%
  \BibitemOpen
  \bibfield  {author} {\bibinfo {author} {\bibfnamefont {F.}~\bibnamefont
  {Bolton}},\ }\href@noop {} {\bibinfo {title} {{QFrame source code}}},\
  \bibinfo {howpublished} {\url{https://github.com/bradan-quantum/qframe/}}
  (\bibinfo {year} {2026}{\natexlab{b}})\BibitemShut {NoStop}%
\bibitem [{\citenamefont {Bolton}(2026{\natexlab{c}})}]{QFrame_PyPI}%
  \BibitemOpen
  \bibfield  {author} {\bibinfo {author} {\bibfnamefont {F.}~\bibnamefont
  {Bolton}},\ }\href@noop {} {\bibinfo {title} {{QFrame package}}},\ \bibinfo
  {howpublished} {\url{https://pypi.org/project/qframe/}} (\bibinfo {year}
  {2026}{\natexlab{c}})\BibitemShut {NoStop}%
\end{thebibliography}%

\end{document}